\documentclass[12pt,a4paper]{article}

\usepackage{amssymb,amsmath}
\usepackage{graphicx,graphics}

\DeclareMathOperator{\tr}{tr} 
\DeclareMathOperator{\trace}{tr} 
\DeclareMathOperator{\esssup}{ess\,sup}

\usepackage[english]{babel}
\usepackage{epsfig,url}
\usepackage{bbm,theorem}
\usepackage{a4wide}

\newtheorem{theorem}{Theorem}[section]

\newtheorem{corollary}[theorem]{Corollary}
\newtheorem{lemma}[theorem]{Lemma}
\newtheorem{proposition}[theorem]{Proposition}
\newtheorem{assumption}[theorem]{Assumption}

{\theorembodyfont{\upshape}

}
\numberwithin{equation}{section}
\numberwithin{theorem}{section}

\newcommand{\qed}{\hfill$\Box$}
\newenvironment{proof}{\begin{trivlist}\item[]{\em Proof:}\/}{%
\qed\end{trivlist}}
\newenvironment{proofof}[1]{%
\begin{trivlist}\item[]{\em Proof of #1}\ }{\qed\end{trivlist}}

\newcommand{\Z}{{\mathbb Z}}
\newcommand{\R}{{\mathbb R}}
\newcommand{\C}{{\mathbb C\hspace{0.05 ex}}}
\newcommand{\N}{{\mathbb N}}
\newcommand{\T}{{\mathbb T}}
\newcommand{\cf}{{\mathbbm 1}} 

\newcommand{\ci}{{\rm i}}
\newcommand{\re}{{\rm Re\,}}
\newcommand{\im}{{\rm Im\,}}
\newcommand{\rme}{{\rm e}}
\newcommand{\rmd}{{\rm d}}

\newcommand{\spr}[2]{\langle #1,#2\rangle}
\newcommand{\FT}[1]{\hat{#1}}

\newcommand{\vc}[1]{\boldsymbol{#1}} 
\DeclareMathOperator*{\sign}{sign}
\DeclareMathOperator*{\supp}{supp}

\newcommand{\norm}[1]{\Vert #1\Vert}
\newcommand{\defset}[2]{ \left\{ #1\left|\,
 #2\makebox[0cm]{$\displaystyle\phantom{#1}$}\right.\!\right\} }
\newcommand{\set}[1]{\{#1\}}

\newcommand{\vep}{\varepsilon}
\newcommand{\defem}[1]{{\em #1\/}}
\newcommand{\qand}{\quad\text{and}\quad}

\newcommand{\Heff}{H_{{\rm eff}}}
\newcommand{\casen}[1]{^{|#1}}
\newcommand{\totcoll}{\sigma_{{\rm coll}}}

\newcommand{\trunc}[1]{\Phi[#1]}
\newcommand{\omt}{\tilde{\Omega}}
\newcommand{\omn}{\Omega}

\newcommand{\Xphys}{X_{\text{ferm}}}
\newcommand{\Lphys}{L^2_{\text{ferm}}}
\newcommand{\Xherm}{X_{\mathbbm{H}}}
\newcommand{\Lherm}{L^2_{\mathbbm{H}}}

\newcommand{\Ctot}{\mathcal{C}}
\newcommand{\Ccoll}{\mathcal{C}_{{\rm diss}}}
\newcommand{\Cdisp}{\mathcal{C}_{{\rm cons}}}
\newcommand{\Closs}{\mathcal{D}}
\newcommand{\Cgain}{\mathcal{G}}
\newcommand{\Ccolltr}{\mathcal{C}_{{\rm diss,tr}}}
\newcommand{\Cdisptr}{\mathcal{C}^\vep_{{\rm cons,tr}}}
\newcommand{\Ctr}{\mathcal{C}^\vep_{{\rm tr}}}

\newcommand{\Ckint}{\int_{(\mathbb{T}^{d})^3}\!\rmd k_2\rmd k_3 \rmd k_4\,}
\newcommand{\sigmapn}[1]{\sigma^{(#1)}}

\newcounter{jlisti}
\newenvironment{jlist}[1][(\thejlisti)]{\begin{list}{{\rm #1}\ \ }{ %
      \usecounter{jlisti} %
    \setlength{\itemsep}{0pt}
    \setlength{\parsep}{0pt}  %
    \setlength{\leftmargin}{0pt} %
    \setlength{\labelwidth}{0pt} %
    \setlength{\labelsep}{0pt} %
}}{\end{list}}

\newcommand{\email}[1]{E-mail: \tt #1}
\newcommand{\emailpeng}{\email{mei@mappi.helsinki.fi}}

\newcommand{\emailjani}{\email{jani.lukkarinen@helsinki.fi}}
\newcommand{\addressjani}{\em University of Helsinki,
Department of Mathematics and Statistics\\
\em P.O. Box 68,
FI-00014 Helsingin yliopisto, Finland}
\newcommand{\addressherbert}{\em Zentrum Mathematik,
Technische Universit\"at M\"unchen, \\
\em Boltzmannstr. 3, D-85747 Garching, Germany}
\newcommand{\emailherbert}{\email{spohn@ma.tum.de}}

\title{Global well-posedness of the spatially homogeneous Hubbard-Boltzmann equation}
\author{Jani Lukkarinen\thanks{\emailjani}, Peng Mei\thanks{\emailpeng},
  Herbert Spohn\thanks{\emailherbert}\\[1em]
$^*, ^\dag$\addressjani \\[1em] $^\ddag$\addressherbert }
\date{\today}

\begin{document}

\selectlanguage{english}
\maketitle

\begin{abstract} 
The Hubbard model is a simplified description for the
evolution of interacting spin-$\frac{1}{2}$ fermions on a $d$-dimensional
lattice.  In a kinetic scaling limit, the Hubbard model can be associated with
a matrix-valued Boltzmann equation, the Hubbard-Boltzmann equation.  Its
collision operator is a sum of two qualitatively different terms: The first
term is similar to the collision operator of the fermionic Boltzmann-Nordheim
equation.  The second term leads to a momentum-dependent rotation of the spin
basis. The rotation is determined by a principal value integral which depends
quadratically on the state of the system and might become singular for
non-smooth states.  In this paper, we prove that the spatially homogeneous
equation nevertheless has global solutions in $L^\infty(\T^d,\C^{2\times2})$
for any initial data $W_0$ which satisfies the ``Fermi constraint'' in the
sense that $0\le W_0 \le 1$ almost everywhere.  We also prove that there is a
unique ``physical'' solution for which the Fermi constraint holds at all
times.  For the proof, we need to make a number of assumptions about the
lattice dispersion relation which, however, are satisfied by the nearest
neighbor Hubbard model, provided that $d\ge 3$.  These assumptions suffice to
guarantee that, although possibly singular, the local rotation term is
generated by a function in $L^2(\T^d,\C^{2\times2})$.
\end{abstract}

\section{Introduction}
\label{sec:intro}

As discovered independently by Nordheim \cite{Nord28} and Peierls
\cite{Peierls29}, the dynamics of weakly interacting quantum fluids can be
well approximated by a kinetic equation of Boltzmann type.  Mathematical
properties of such quantum kinetic equations have been studied extensively,
for a general review one can consult \cite{VillaniRev} and for recent results
we refer to \cite{dolb94,Lions94iii,Lu04,Lu05,EMV08,EV12a,EV12b}.  Of
particular interest in our context is the work of Dolbeault \cite{dolb94} on
the Boltzmann-Nordheim equation for spinless fermions.

In our contribution, we will study the kinetic equation derived from the
fermionic Hubbard model, see \cite{flms12} for details.  Compared to the
kinetic equation of \cite{dolb94}, there are three important modifications.
The Hubbard model describes the motion of electrons in a periodic background
potential in the tight binding approximation, which means that the electrons
move on the $d$-dimensional lattice $\Z^d$. In the spatially homogeneous case,
the one considered here, this implies that the Wigner function $W(t,k)$ at
time $t$ is a function on the $d$-torus, $k\in \T^d$, the fundamental zone for
the discrete Fourier transform.  Secondly, since electrons have spin
$\frac{1}{2}$, the Wigner function $W(t,k)$ depends on the spin and so
naturally forms a $2\times 2$-matrix.  Thus the kinetic equation governs the
evolution of a matrix-valued function on the torus for which, in general,
$[W(t,k_1),W(t,k_2)]\ne 0$ if $k_1\ne k_2$.  Hence the conventional arguments
have to be reworked.

The third modification is an appearance of a Vlasov type term in the Boltzmann
equation.  In the Hubbard model, the electron interaction is on-site and
independent of spin.  Hence, the microscopic Hamiltonian is invariant under
global spin rotations.  As a consequence, besides a conventional collision
operator, the Hubbard-Boltzmann equation contains a term similar to that of
the Vlasov equation.  The new term rotates the $k$-dependent spin basis but
does not generate entropy.  At the first sight, the Vlasov term appears
innocuous but, in fact, it is one major obstacle to overcome before arriving
at a well-posed quantum kinetic equation: the term is defined by a principal
value integral which might generate singularities even for regular initial
data.

In the spatially homogeneous case, the Hubbard-Boltzmann equation reads
\begin{align}\label{eq:HBE}
 \partial_t W(t,k) = \Ctot [W(t,\cdot)](k)\, ,\qquad 
\Ctot[W]:= \Ccoll [W]+ \Cdisp [W]\, .
\end{align} Here $W=W_{ij}(t,k)$ is a $2 \times 2$ Hermitian matrix at time
$t\ge 0$ with wave number $k\in\mathbb{T}^d$ and $i,j\in\set{1,2}$ counting
for the spin-degrees of freedom.  Physically, $W$ describes the average
density of electrons at wave number $k$ with a spin density matrix $W(t,k)/\tr
(W(t,k))$.  We wish to solve this equation subject to an initial condition
\begin{align}
\label{initial} W(0,k)=W_0(k)\, ,  \qquad \text{for all } k \in \mathbb{T}^d\,
,
\end{align} i.e., as a Cauchy problem with initial data $W_0$.  The collision
operator $\Ctot[W]$ is a sum of a dissipative term $\Ccoll [W]$ and a
conservative Vlasov type term $\Cdisp[W]$.  Their properties depend crucially
on the dispersion relation $\omega:\T^d\to \R$ of the waves generated by the
quadratic term of the tight-binding Hamiltonian.  For instance, in the Hubbard
model on a $d$-dimensional square lattice, $\omega(k)= - \sum_{\nu=1}^d \cos
(2\pi k^\nu)$.

The \defem{dissipative part} of the collision operator is given by
\begin{align}
\label{eq:collop2}
& \Ccoll [W](k_1):=\pi\Ckint \delta(\underline{k}) \delta(\underline{\omega}) 
\nonumber \\  & \quad 
 \times  \left( \tilde{W}_1 W_3 J[\tilde{W}_2 W_4]
+ J[ W_4 \tilde{W}_2] W_3 \tilde{W}_1 
 -W_1\tilde{W}_3 J[W_2\tilde{W}_4]
- J[\tilde{W}_4 W_2]\tilde{W}_3  W_1 \right) \, ,
\end{align}
where we employ the following notations
\begin{align}
  &\underline{k} := k_1+k_2-k_3-k_4\, ,\quad
  \underline{\omega} := \omega_1+\omega_2-\omega_3-\omega_4 \, ,\quad
  \tilde{W} := 1 - W \, ,\\
  & \omega_i := \omega(k_i)\, ,\quad W_i := W(k_i) \, ,\quad \tilde{W}_i := \tilde{W}(k_i) 
  \, , \quad \text{for } i=1,2,3,4\, , \\
  & J[M] := 1 \tr (M) - M,\ \text{ for } M\in \C^{2\times 2}\, . \label{eq:defJ}
\end{align}
The shorthand notations $\underline{k}$ and $\underline{\omega}$ are somewhat
rigid, as the dependence on variables $k_i$, $i=1,2,3,4$, is implicit.  We
will only use them to shorten the notation for the collision kernels, in which
case $k_1$ denotes always the fixed ``input'' variable and $k_i$, $i=2,3,4$,
the integration variables.  The \defem{conservative part} of the collision
operator is written as a commutator with a $W$-dependent ``effective
Hamiltonian'', 
\begin{align}\label{eq:defCdisp} 
& \Cdisp[W](k_1) := -\ci [\Heff[W] (k_1),W(k_1)]\, ,
\end{align}
where $\Heff$ is defined formally as a principal value integral around
$\underline{\omega}=0$, 
\begin{align}\label{eq:defHeff}
& \Heff [W](k_1) := -\frac{1}{2}\,  
\text{p.v.}\!\Ckint \delta(\underline{k})\, \frac{1}{\underline{\omega}} 
\, \nonumber \\  & \quad \times
 \left( W_3 J[\tilde{W}_2 W_4]
+ J[ W_4 \tilde{W}_2] W_3 + 
\tilde{W}_3 J[W_2\tilde{W}_4] + J[\tilde{W}_4 W_2]\tilde{W}_3  \right) \, .
\end{align}

The goal of our contribution is to establish that the evolution equation
(\ref{eq:HBE}), together with  
(\ref{initial}), (\ref{eq:collop2}), (\ref{eq:defCdisp}) and
(\ref{eq:defHeff}), is well-posed for ``fermionic initial data''.  Since the
original Hubbard model describes fermions, the above Wigner matrix function at
any time $t$ needs to satisfy $0\le W(t,k)\le 1$, as a matrix
inequality\footnote{We follow the convention that $M\ge 0$ if and only if $M$
  is a Hermitian matrix and its eigenvalues belong to $[0,\infty)$.  We also
  recall that $M\in \C^{n\times n}$ satisfies $M\ge 0$ if and only if $(z,M
  z)\ge 0$ for all $z\in \C^n$, and that this result is not valid if it is
  only checked for $z\in \R^n$.} for almost every $k$; we call this property
the \defem{Fermi constraint}.   
Thus from the physics side, it is natural to look for solutions in the space
of Lebesgue measurable functions $W(k)$ which satisfy the Fermi constraint
almost everywhere.   This requires to show that, if the initial data $W_0$
satisfies $0\le W_0(k)\le 1$, then this property is propagated in time.  For
this purpose we use the approach of Dolbeault, somewhat modified to account
for the matrix-valuedness of $M$ and of the above constraints.  Our
construction relies heavily on the property that  
\begin{align}\label{eq:mainmatrixbound}
  A J[BC] + C J[B A] \ge 0\, ,
\end{align}
for arbitrary $n\times n$ matrices $A,B,C\ge 0$.  (We thank David Reeb for
most helpful discussions relating to the inequality
(\ref{eq:mainmatrixbound}).) 

For generic $W\in L^\infty(\T^d,\C^{2\times 2})$ we cannot expect that the
principal value integral in (\ref{eq:defHeff}) is convergent for all $k_1\in
\T^d$, or even that it converges almost everywhere to a bounded function.
(Already the standard Hilbert-transform---a unitary operator on $L^2(\R)$,
defined via a similar principal value integral---offers such examples: the
Hilbert transform of the characteristic function of any bounded interval has
logarithmic divergences at both ends of the interval.)  Of course, one could
try to restrict the study to more regular function spaces in which the limit
exists pointwise everywhere.  However, there is no a priori reason why such a
space would be invariant under the time-evolution.  Our actual construction is
rather indirect, but does achieve the desired goal.  More precisely, our
strategy is to prove first  the well-posedness of a regularized problem for
functions continuous in $k$, to solve $\partial_t W = \Ccoll [W]-\ci
[\Heff^\vep[W],W]$, and then to show that the solutions converge in $L^2$-norm
to a unique solution of the original problem as the regulator is removed,
$\vep\to 0^+$. 

In the well-studied continuum setup of Boltzmann equations, the dispersion
relation is given by $k^2$, $k\in \R^d$, and the energy constraint
$\underline{\omega}=0$ has simple explicit solutions and one can integrate
over both $\delta$-functions to obtain an explicit integral operator involving
only Lebesgue measures.  In contrast, all lattice systems share the difficulty
that even for the simplest lattice dispersion relations the solutions to
$\underline{\omega}=0$ are not easy to handle. Integration over
$\delta(\underline{\omega})$ is even more problematic since the result
typically involves singular integral kernels or might become ill-defined.  A
complete classification of the singularities resulting from such an
integration over the $\delta$-functions appears to be a hard problem in
harmonic analysis and is certainly beyond the scope of the present paper.  

Instead of a classification result, we provide here a set of conditions for the
dispersion relation $\omega$ under which both the dissipative collision
integral and the principal value integral defining the effective Hamiltonian
are sufficiently regular for the resulting solutions to be $L^2$-continuous.
These assumptions and the main results are described in Sec.~\ref{sec:main}.
In fact, for continuous $W$ the dissipative part $\Ccoll$ can be defined even
without the last (and the most complicated) of these conditions.    We prove
in Sec.~\ref{sec:Ccoll} that then the formal integrations over the
$\delta$-functions in (\ref{eq:collop2}) can be replaced by an integration
over a family of naturally defined bounded Borel measures.  These measures are
used in Sec.~\ref{sec:regIPV} to prove that the regularized problem is
well-posed for continuous initial data.  We have presented these results in
somewhat greater generality than what is needed for the proof of the main
theorem: such measures appear also in other phonon Boltzmann equations, and
the properties proven in  Sec.~\ref{sec:Ccoll} could thus be of independent
interest.   

In Sec.~\ref{sec:pv}, we show how the full set of assumptions leads to
$L^2$-continuity of the collision operator and apply the resulting estimates
to conclude the proof of the main theorem in Sec.~\ref{sec:provemain}.  One
technical problem in extending the phonon Boltzmann collision operator from
continuous $W$ to functions defined only Lebesgue almost everywhere is related
to the measures derived in Sec.~\ref{sec:Ccoll} which are singular with
respect to the Lebesgue measure.  Therefore, one needs to handle sets of
measure zero carefully.  In fact, it turns out that instead of working
directly with integrals over the above Borel measures, it is better to define
the collision integral via a limit procedure using $L^2$-approximants which
are continuous in $k$ (Lemma \ref{th:l2colloper0} and Corollary
\ref{th:l2colloper}).   These two definitions need not be pointwise
equivalent, as the following example illustrates: Consider a function $f(x)$
which takes value one at $x=0$ and is zero elsewhere.  Then $\int\!\rmd x\,
\delta(x-x_0) f(x) = 1$ at $x_0=0$ but $f$ is $L^2$-equal to the zero map
$f_0$ for which $\int\! \rmd x\, \delta(x-x_0) f_0(x) = 0$ for all $x_0$.
However, should the two definitions of the collision operator disagree with
each other, then the one presented here looks physically more reasonable: for
instance, we prove here that it results in global well-posedness and in
conservation of both total energy and spin.  In fact, we show in Proposition
\ref{th:simpleccoll} that the above definition of the collision integral can
be simplified for dispersion relations satisfying all of the present
assumptions. Then it can be defined for $W\in L^\infty$ as an $L^2$-limit of
regularized Lebesgue integrals similar to those used in the definition of the
principal value integral in $\Heff$.

To check that a given dispersion relation satisfies the required assumptions,
is still a nontrivial problem in itself.  We have included a proof in Appendix
\ref{sec:dispnn} which implies that our results indeed apply to the nearest
neighbor square lattice hopping, at least as long as the dimension is high
enough, $d\ge 3$.

\subsection*{Acknowledgements}

The research of J.~Lukkarinen and P.~Mei was supported by the Academy of
Finland and partially by the ERC Advanced Investigator Grant 227772.  We 
are also grateful to the Nordic Institute for Theoretical Physics
(NORDITA), Stockholm, Sweden, to the Erwin Schr\"odinger International
Institute for Mathematical Physics (ESI), Vienna, Austria,  and to the Banff
International Research Station for Mathematical Innovation and Discovery
(BIRS), Banff, Canada, for their hospitality during the workshops in which
part of the research for the present work has been performed.  H.~Spohn thanks
M.~F\"urst and C.~Mendl for helpful discussions.

\section{Main results}
\label{sec:main}

To give a proper mathematical definition of the collision operator $\Ctot
[W]$, we need to add some assumptions about the dispersion relation $\omega$.
For instance, if we would allow $\omega$ to be a constant map, we would have
$\underline{\omega}=0$ everywhere and thus $\delta(\underline{\omega})=\infty$
identically.  As in \cite{NLS09}, we formulate our assumptions through
properties of oscillatory integrals involving $\omega$.  There is considerable
freedom, and the following choices are mainly made for convenience:
(DR\ref{it:DRdisp}) will imply that the term $\delta(\underline{\omega})$
leads to uniformly bounded measures and (DR\ref{it:DRpv}) will be used to
prove the $L^2$-continuity of $\Heff$.  Neither of these conditions is likely
to be optimal, but they facilitate the technical estimates and are general
enough to include many physically relevant cases. For instance, in Appendix
\ref{sec:hubbd3} we show that nearest neighbor interactions on a square
lattice with $d\ge 3$ satisfy all of the  assumptions.

Whenever necessary, the $d$-torus $\T^d:=\R^d/\Z^d$ is understood as having
the parameterization $[-\frac{1}{2},\frac{1}{2}]^d$ with periodic
identifications on the boundary.  In particular, arithmetic is periodically
extended to $\T^d$ and we have a normalization $\int_{\T^d} \!\rmd k=1$.

\begin{assumption}[Dispersion relation]\label{th:disprelass}
Suppose $d\ge 1$ and $\omega:\T^d\to\R$ satisfies all of the following:
\begin{jlist}[(DR\thejlisti)]
\item\label{it:DR1} The periodic extension of $\omega$ is continuous
and satisfies $\omega(-k)=\omega(k)$.
\item\label{it:DRdisp} ($\ell_3$-dispersivity). Let us consider the 
{\em free propagator\/}
  \begin{align}\label{eq:defptx}
    p_t(x) = \int_{\T^d} \!\rmd k\, \rme^{\ci 2\pi x\cdot k}
    \rme^{-\ci t \omega(k)} \, .
  \end{align}
We assume that its $\ell_3$-norm belongs to $L^3(\rmd t)$, in other words, that
$t\mapsto \norm{p_t}_3^3\in L^1$. 
\item \label{it:DRpv} (effective collisional dispersivity)  For 
$\sigma\in\set{-1,1}^4$ and $k,k'\in (\T^d)^3$ define first 
\begin{align}\label{eq:defomsig}
  \omt (k;\sigma) := 
  \sum_{i=1}^3 \sigma_i \omega(k_i) + \sigma_4 \omega(k_1+k_2-k_3)
\end{align}
and then
\begin{align*}
 & \Omega_1(k,k';\sigma) := \omt(k;\sigma) \, , \qquad
   \Omega_2(k,k';\sigma) := \omt((k_1,k'_2,k_3);\sigma) \, ,\\ &
  \Omega_3(k,k';\sigma) := \omt(k';\sigma)  \, , \qquad
   \Omega_4(k,k';\sigma) := \omt((k'_1,k_2,k'_3);\sigma)\, .
\end{align*}
Set for $s\in \R^4$, $\sigma\in \set{-1,1}^4$,
\begin{align}
  \mathcal{G}(s;\sigma) := \int_{(\T^d)^3\times (\T^d)^3} \!\rmd^3 k' 
\rmd^3 k\, \rme^{\ci \sum_{i=1}^4 s_i \Omega_i(k,k';\sigma)}\, .
\end{align}
We assume that $C_{\mathcal{G}}:= \max_\sigma \int_{\R^4} \rmd s
|\mathcal{G}(s;\sigma)| < \infty$. 
\end{jlist}
\end{assumption}

To study the solutions we consider the following function spaces
\begin{align}
  \Xherm & := \defset{W\in C(\T^d, \C^{2\times 2})}{
    W(k)^* = W(k)\text{ for all } k\in \T^d}\, , \\
  \Xphys & := \defset{W\in C(\T^d, \C^{2\times 2})}{ 
    0\le W(k)\le 1\text{ for all } k\in \T^d}\, ,\\
  \Lherm & := \defset{W\in L^2(\T^d, \C^{2\times 2})}{ 
    W(k)^* = W(k)\text{ for a.e.\ } k\in \T^d}\, ,\\
  \Lphys & := \defset{W\in L^2(\T^d, \C^{2\times 2})}{ 
    0\le W(k)\le 1\text{ for a.e.\ } k\in \T^d}\, .
\end{align} We equip the function space $\Xherm$ with the $\sup$-norm and the
equivalence classes in $\Lherm$ with $L^2$-norm which makes both spaces into
\defem{real} Banach spaces.  It is easy to check that then $\Xphys$ and
$\Lphys$ are closed convex subsets of $\Xherm$ and $\Lherm$, respectively.
  
To make the definitions of the norms explicit, we need to fix the matrix norm
for $\C^{2\times 2}$.  Since any two norms on a finite-dimensional vector
space are equivalent, the choice does not play much role in the results.
However, it will be convenient to consider here the Hilbert-Schmidt norm: we
define for $M\in \C^{2\times 2}$ 
\begin{align} \norm{M}^2 := \sum_{i,j=1}^2 |M_{ij}|^2\, .
\end{align} For the matrix product this implies an estimate $\norm{M' M}\le
\norm{M'} \norm{M} $.  The resulting $L^2$- and $L^\infty$-norms in $\Lherm$
are denoted by $\norm{W}_2$ and $\norm{W}_\infty$; explicitly, $\norm{W}_2^2 =
\int_{\T^d}\!\rmd k\, \tr(W(k)^*W(k))$ and $\norm{W}_\infty = \esssup_{k\in
\T^d} \sqrt{\tr(W(k)^*W(k))}$, where the essential supremum refers to the
Lebesgue measure.

Our main results are summarized in the following three theorems.  The first
two give a precise meaning to the formal notations used in the definition of
the collision operator.  In the first theorem we explain how $\Ccoll$ is
connected to a natural bounded Borel measure.
\begin{theorem}\label{th:maindefC} Suppose $\omega$ satisfies the conditions
(DR\ref{it:DR1}) and  (DR\ref{it:DRdisp}).  In the collision operator, for a
given $k_1\in \T^d$, the notation $\rmd k_2\rmd k_3 \rmd k_4\,
\delta(\underline{k})\delta(\underline{\omega})$ refers to a regular complete
bounded positive measure on $(\T^d)^3$ whose $\sigma$-algebra contains Borel
sets and which satisfies for any $F\in C((\T^d)^3)$ 
\begin{align} 
\Ckint \delta(\underline{k})\delta(\underline{\omega})
F(k_2,k_3,k_4) =  \lim_{\vep \to 0^+} 
\int_{(\T^d)^2} \frac{\rmd k_2 \rmd k_3}{\pi} \,  F(k_2,k_3,k_1+k_2-k_3)
\frac{\vep}{\vep^2+\underline{\omega}^2}\, .
\end{align} 
If $W\in \Xphys$, then $\Ccoll [W]$ is defined using this measure
in (\ref{eq:collop2}).

For $W\in \Lphys$, $\Ccoll [W]\in \Lherm$ is defined by using continuous
approximants, as explained in detail in Corollary \ref{th:l2colloper}.   If
all of the conditions (DR\ref{it:DR1})--(DR\ref{it:DRpv}) hold, then for any
$W\in \Lphys$
\begin{align}\label{eq:collop4} 
& \Ccoll [W](k_1) = \lim_{\vep\to 0^+}  \Ckint
\delta(\underline{k})\, \frac{\vep}{\vep^2+\underline{\omega}^2} 
\nonumber \\ & \quad  \times  
\left( \tilde{W}_1 W_3 J[\tilde{W}_2 W_4] + 
J[ W_4 \tilde{W}_2] W_3 \tilde{W}_1  -W_1\tilde{W}_3 J[W_2\tilde{W}_4] -
J[\tilde{W}_4 W_2]\tilde{W}_3  W_1 \right) \, ,
\end{align} 
where the limit is taken in $\Lherm$-norm.
\end{theorem}

The second result shows that $\Heff$ can be defined as an ``$L^2$-principal
value integral'' which can also be obtained as a limit of terms using more
regular cutoffs.  (In fact, the second regularization arises naturally in the
derivation of the equation from the microscopic Hubbard model, cf.\
\cite{flms12}.)
\begin{theorem}\label{th:maindefHeff} Suppose $\omega$ satisfies all of the
conditions (DR\ref{it:DR1})--(DR\ref{it:DRpv}) and define for $\vep>0$
\begin{align}\label{eq:defHeffeps} & \Heff ^\vep[W](k_1) := -\frac{1}{2} \Ckint
\delta(\underline{k})\, \frac{\underline{\omega}}{ \underline{\omega}^2 +
\vep^2} \, \nonumber \\  & \quad \times  \left( W_3 J[\tilde{W}_2 W_4] + J[
W_4 \tilde{W}_2] W_3  + \tilde{W}_3 J[W_2\tilde{W}_4] + J[\tilde{W}_4
W_2]\tilde{W}_3  \right) \, .
\end{align}  Then, given $W\in \Lphys$ there is an $\Lherm$-limit of
$\Heff^\vep[W]$ as $\vep\to 0^+$ which we denote by $\Heff[W]$.  In addition,
for any $W\in \Lphys$ we can find a sequence $\vep_n\to 0^+$ such that for
almost every $k_1$
\begin{align}\label{eq:L2pv}  & \Heff[W](k_1) = \lim_{n\to \infty}
-\frac{1}{2}
\Ckint \delta(\underline{k})\, \frac{\cf(|\underline{\omega}|\ge
\vep_n)}{\underline{\omega}} \,  \nonumber \\  & \quad \times  \left( W_3
J[\tilde{W}_2 W_4] + J[ W_4 \tilde{W}_2] W_3  + \tilde{W}_3 J[W_2\tilde{W}_4]
+ J[\tilde{W}_4 W_2]\tilde{W}_3  \right) \, ,
\end{align}  where $W_i:=W(k_i)$, $i=1,2,3,4$.
\end{theorem}

The third theorem shows that, endowed with the above definitions, the
Hubbard-Boltzmann equation is well-posed for any fermionic initial data.  It
also implies that the solution satisfies the basic properties expected from a
kinetic scaling limit of the original Hubbard model: preservation of the Fermi
property and of total energy and spin.
\begin{theorem}\label{th:main1} Suppose $\omega$ satisfies all of the
conditions (DR\ref{it:DR1})--(DR\ref{it:DRpv}) and define $\Ccoll$ and $\Heff$
as in Theorems \ref{th:maindefC} and \ref{th:maindefHeff}, respectively.  If
$W_0\in \Lphys$, then there is a unique $W\in C^{(1)}([0,\infty),\Lphys)$ such
that $W(0,k)=W_0(k)$ almost everywhere and for all $t>0$ 
\begin{align}
  \partial_t W_t = \Ccoll [W_t] - \ci [\Heff[W_t],W_t]\, ,
\end{align} where $W_t(k):=W(t,k)$.  Here $W_t$ depends $\Lherm$-continuously
on $W_0$ and the dependence is uniform on any compact interval of
$[0,\infty)$.  In addition, for all $t\ge 0$
\begin{align} & \int_{\T^d}\!\rmd k\, \omega(k) \trace W(t,k) =
\int_{\T^d}\!\rmd k\, \omega(k) \trace W_0(k) \, ,\label{eq:claw1} \\ &
\int_{\T^d}\!\rmd k\, W(t,k) = \int_{\T^d}\!\rmd k\, W_0(k) \in \C^{2\times
2}\, . \label{eq:claw2}
\end{align}
\end{theorem}

Here the notation $C^{(1)}([0,\infty),\Lphys)$ denotes continuous functions
$[0,\infty)\to \Lphys\subset \Lherm$ which are continuously Fr\'echet
differentiable on $(0,\infty)$, considered as an open subset of the Banach
space $\R$, assuming also that the limit $t\to 0^+$ of the derivatives exists.
We recall here the basic property that if $W\in C^{(1)}([0,\infty),\Lphys)$
then for all $0\le t_0\le t_1 <\infty$
\begin{align}\label{eq:L2hermcalc} 
W_{t_1} = W_{t_0} + \int_{t_0}^{t_1} \!\rmd s\, \partial_s W_s\, ,
\end{align} 
where the integral is understood as a vector valued integral in
$\Lherm$ over the compact set $[t_0,t_1]$ and $\partial_s W_s \in \Lherm$
denotes the Fr\'echet derivative at $s\in [t_0,t_1]$. (Strictly speaking, the
Fr\'echet derivative should here be an element $D\in \mathcal{B}(\R,\Lherm)$.
However, the action of this map is uniquely determined by giving $D 1\in
\Lherm$ and we always identify $D$ with $D 1$ for such maps and denote this by
$\partial_s W_s$.  Also, such a map is continuously Fr\'echet differentiable
if and only if the map $s\mapsto \partial_s W_s$ is continuous which shows
that the above vector valued integral in $\Lherm$ is well-defined for any
$W\in C^{(1)}([0,\infty),\Lphys)$.  After these preliminaries,
(\ref{eq:L2hermcalc}) follows from the fundamental theorem of calculus and the
easily checked property that $\partial_s \Lambda[W_s]=\Lambda[\partial_s W_s]$
for any dual element $\Lambda\in (\Lherm)^*$.)

\section{Basic properties of the dissipative term $\Ccoll$}
\label{sec:Ccoll}

This section concerns the definition of the dissipative part of the collision
operator.  In particular, our goal is to show that the first two assumptions,
(DR\ref{it:DR1}) and (DR\ref{it:DRdisp}), suffice to define the dissipative
term $\Ccoll[W]$ via a Borel measure with an $L^2$-continuity property which
will be needed later.

We begin by showing that the measure merits its symbolic notation which uses
the two formal delta-functions, at least as long as the integrand is
continuous.  This is the main goal of the first Proposition here.  The result
is somewhat more general than what we need for the present proof but the
apparently superfluous properties could well become useful in rigorous studies
of other phonon Boltzmann equations.  For the statement, we consider more
general ``energy constraints'': define for $k\in (\T^d)^4$ and 
$\sigma \in \set{-1,1}^4$
\begin{align}\label{eq:defomk} 
& \omn(k,\sigma):=\sum_{i=1}^4 \sigma_i \omega(k_{i}) \, .
\end{align} 
This is related to the definition in (\ref{eq:defomsig}) by
$\omn(k,\sigma)=\omt((k_1,k_1,k_2),\sigma)$, whenever $k_1+k_2=k_3+k_4$ (mod
$1$).  In addition, the combination appearing in the definition of the collision
operator satisfies $\underline{\omega}=\omn(k,(1,1,-1,-1))$.

\begin{proposition}\label{th:defC1}
  Assume that $\omega:\T^d\to \R$ is continuous and satisfies
  (DR\ref{it:DRdisp}), and define $\omt$ and $\omn$ as in (\ref{eq:defomsig})
  and (\ref{eq:defomk}).  Then for every $k_1\in \T^d$, $\alpha\in \R$, and
  $\sigma \in \set{-1,1}^4$ the map  
\begin{align}\label{eq:numaindef}
  C((\T^d)^3)\ni F \mapsto \lim_{\vep \to 0^+} \int_{(\T^d)^2} 
  \frac{\rmd k_2\rmd k_3}{\pi} \,  F(k_2,k_3,k_1+k_2-k_3) 
  \frac{\vep}{\vep^2+(\omt((k_1,k_2,k_3),\sigma)-\alpha)^2}\, ,
\end{align}
defines a regular complete bounded positive measure on $(\T^d)^3$,
$\sigma$-algebra containing Borel sets, which we denote by
$\nu_{k_1,\alpha,\sigma}(\rmd^3 k)$ or 
$\Ckint \delta(\omn(k,\sigma)-\alpha)\delta(k_1+k_2-k_3-k_4)$. 

The measure has the following properties, for any fixed $\sigma$:
\begin{enumerate}
 \item\label{it:totcoll} Define $\totcoll(k_1,\alpha;\sigma):=\int\! \rmd
   \nu_{k_1,\alpha,\sigma}$.  Then $\totcoll\in C(\T^d\times\R)$. 
 \item\label{it:support} 
The set $\defset{(k_2,k_3,k_4)\in (\T^d)^3}{k_1+k_2-k_3-k_4=0,\
  \Omega(k,\sigma)=\alpha}$ contains the support of the measure
$\nu_{k_1,\alpha,\sigma}$. 
 \item\label{it:fomint} 
 Suppose $F\in C((\T^d)^3)$, $f\in C(\R)$, $k_1\in \T^d$, and $\alpha\in \R$
 are given.  Then 
 \begin{align}
   \int\!\nu_{k_1,\alpha,\sigma}(\rmd^3 k')\, f(\Omega((k_1,k'),\sigma)) F(k')
   =  f(\alpha)  \int\!\nu_{k_1,\alpha,\sigma}(\rmd^3 k')\, F(k')\, .
 \end{align}
In particular, 
$\int\!\nu_{k_1,0,\sigma}(\rmd^3 k')\, \Omega((k_1,k'),\sigma)F(k') =0$ 
for any $F\in C((\T^d)^3)$ and all $k_1\in \T^d$. 
\item\label{it:contin} 
 If $G\in C((\T^d)^4\times\R)$, then the map $(k_1,\alpha) \mapsto
 \int\!\nu_{k_1,\alpha,\sigma}(\rmd^3 k')\, G(k_1,k',\alpha)$ is continuous on
 $\T^d\times\R$. 
\item\label{it:alphaprop}  
 Suppose $F\in C((\T^d)^3\times \R)$ and $m>\sup_k |\omt(k,\sigma)|$.  Then
 for all $k_1\in \T^d$ we have $\int_{-m}^m\!\rmd
 \alpha\left(\int\!\nu_{k_1,\alpha,\sigma}(\rmd^3 k')\, F(k',\alpha)\right)=
 \int_{(\T^d)^2}\!\rmd k_2 \rmd k_3 \,
 F(k_2,k_3,k_1+k_2-k_3,\omt((k_1,k_2,k_3),\sigma))$. 
\end{enumerate}
\end{proposition}
\begin{proof} $\sigma$ will be considered fixed in the following, and we drop
  the dependence on it from the notation here. 
  Define for $x\in \R$ and $\vep>0$
\begin{align}\label{eq:defvarphi0}
  \varphi^0_\vep(x):= \frac{\vep}{x^2+\vep^2}\qand
 \FT{\varphi}^0_\vep(x) := \frac{1}{2} \rme^{-\vep |x|}\, .
\end{align}
Then $\FT{\varphi}^0_\vep \in L^1\cap L^\infty$, and
it is straightforward to check that pointwise for {\em all\/} $x\in \R$ we have 
$\varphi^0_\vep(x) = \int_{-\infty}^\infty \!\rmd s\, 
\FT{\varphi}^0_\vep(s)\rme^{\ci s x}$.  We need to consider the maps 
\begin{align}
  \Lambda_{k_1,\alpha,\vep}[F] := \int_{(\T^d)^2} 
  \frac{\rmd k_2 \rmd k_3}{\pi} \,  F(k_2,k_3,k_1+k_2-k_3)
  \varphi^0_\vep(\omt((k_1,k_2,k_3),\sigma)-\alpha) 
\end{align}
which are all  obviously  continuous linear functionals on $C((\T^d)^3)$ with
\begin{align}
  \norm{\Lambda_{k_1,\alpha,\vep}} \le \int_{(\T^d)^2} \frac{\rmd k_2 \rmd
    k_3}{\pi} \, \varphi^0_\vep(\omt((k_1,k_2,k_3),\sigma)-\alpha) =:
  C_{k_1,\alpha,\vep}\, . 
\end{align}
In addition, Fubini's theorem yields for any $F\in C((\T^d)^3)$
\begin{align}
  \Lambda_{k_1,\alpha,\vep}[F] = \int_{-\infty}^\infty \!\rmd s\,
  \FT{\varphi}^0_\vep(s) 
\int_{(\T^d)^2} \frac{\rmd k_2 \rmd k_3}{\pi} \,  F(k_2,k_3,k_1+k_2-k_3)
\rme^{\ci s \omt((k_1,k_2,k_3),\sigma)-\ci s \alpha}\, . 
\end{align}

Assume first that $F$ is a trigonometric polynomial, i.e., that there is
$f:(\Z^d)^3 \to \C$ which has finite support and 
\begin{align}
F(k) = \sum_{x\in (\Z^d)^3} \rme^{-\ci 2\pi k\cdot x} f(x)\, .
\end{align}
Then
\begin{align}\label{eq:Fom}
&\int_{(\T^d)^2} \!\rmd k_2 \rmd k_3 \, F(k_2,k_3,k_1+k_2-k_3) \rme^{\ci s
  \omt((k_1,k_2,k_3),\sigma)-\ci s \alpha} 
\nonumber \\ & \quad
= \sum_{x\in (\Z^d)^3} \rme^{\ci s \sigma_1\omega(k_1)-\ci s \alpha} f(x)
\int_{(\T^d)^2} \!\rmd k_2 \rmd k_3 \, 
\rme^{\ci s (\sigma_2\omega(k_2)+\sigma_3\omega(k_3)+\sigma_4\omega(k_1+k_2-k_3))} 
\nonumber \\ & \qquad \times
\rme^{-\ci 2\pi (k_2\cdot (x_1+x_3)+k_3\cdot (x_2-x_3)+k_1\cdot x_3)}\, .
\end{align}
The remaining convolution integral can be expressed in terms of
$p_t(x)$, which is the inverse Fourier transform of
$k\mapsto\rme^{-\ci t \omega(k)}$.  Using Parseval's theorem to the
$k_3$-integral and then Fubini's theorem proves that (\ref{eq:Fom}) is equal
to  
\begin{align}\label{eq:convsplit}
& \sum_{x\in (\Z^d)^3} \rme^{\ci s \sigma_1\omega(k_1)-\ci s \alpha} f(x)
\sum_{y\in \Z^d}
\rme^{-\ci 2\pi k_1\cdot (y+x_2)}  p_{-\sigma_2 s}(-y-x_1-x_2)
p_{-\sigma_4 s}(y+x_2-x_3) p_{-\sigma_3 s}(y)\, .
\end{align}
Thus by H\"{o}lder's inequality, the property
$\norm{p_{-s}}_3 = \norm{p_{s}}_3$, and using (DR\ref{it:DRdisp}),
its absolute value is bounded by
$\norm{f}_1\norm{p_{s}}_3^3\in L^1(\rmd s)$.  Therefore, dominated convergence
can be used here to prove that, when $\vep\to 0^+$, 
\begin{align}\label{eq:trigconv}
  \Lambda_{k_1,\alpha,\vep}[F] \to 
\int_{-\infty}^\infty \!\rmd s \left(
\int_{(\T^d)^2} \frac{\rmd k_2 \rmd k_3}{2 \pi} \, \,  F(k_2,k_3,k_1+k_2-k_3)
\rme^{\ci s \omt((k_1,k_2,k_3),\sigma)-\ci s \alpha}\right)\, . 
\end{align}
In addition, dominated convergence also implies that the limit defines a
continuous function of $(k_1,\alpha)$, as both $F$ and $\omt$ are continuous. 
Applying the bounds for $f(x)= \cf(x=0)$, i.e., for $F(k)=1$, also proves that
there is $C'<\infty$, independent of $\alpha$, $k_1$ and $\vep$, such that  
$C_{k_1,\alpha,\vep} \le C'$.

Therefore, we have now proven that $(\Lambda_{k_1,\alpha,\vep})_{\vep>0}$,
with $\alpha$, $k_1$ fixed, form an equicontinuous family of linear
functionals on $C((\T^d)^3)$ which converges at any $F$ which is a
trigonometric polynomial.  Since trigonometric polynomials are dense in
$C((\T^d)^3)$ by the Stone-Weierstrass theorem, this implies that the family
converges then for any $F\in C((\T^d)^3)$ and the limit defines a unique
$\Lambda_{k_1,\alpha}\in C((\T^d)^3)^*$. 
(Such a statement is true for any equicontinuous sequence of linear mappings
between Banach spaces; for a more generic statement, see for instance
\cite[Exercise 2.14]{rudin:fa}.  By (\ref{eq:trigconv}), the limits obtained
from an arbitrary sequence $\vep_n>0$ with $\vep_n\to 0$ agree on a dense set,
and therefore they must all be equal.  We denote the common limit by
$\Lambda_{k_1,\alpha}$ and it follows that $\lim_{\vep\to 0^+}
\Lambda_{k_1,\alpha,\vep}[F] = \Lambda_{k_1,\alpha}[F]$ for all continuous
$F$, even though it could well happen that the integral on the right hand side
of (\ref{eq:trigconv}) is not absolutely convergent for some $F$.) 
Since obviously $\Lambda_{k_1,\alpha}[F]\ge 0$ for any $F\ge 0$ and $(\T^d)^3$
is compact, Riesz representation theorem implies that there is a complete
regular positive measure $\nu_{k_1,\alpha}$ such that its $\sigma$-algebra
contains Borel sets and 
$\Lambda_{k_1,\alpha}[F]= \int\! \rmd\nu_{k_1,\alpha} F$ 
for $F\in C((\T^d)^3)$.   
In addition, since $\totcoll(k_1,\alpha)=\int\!\rmd\nu_{k_1,\alpha} =
\lim_{n\to \infty} C_{k_1,\alpha,1/n}\le C'$, the measure is bounded and
$\totcoll$ is a continuous function of $k_1,\alpha$. 
We have thus proven the first part of the Proposition, and item
\ref{it:totcoll} in the second part. 

To prove item \ref{it:support}, fix $k_1$, denote $\nu:=\nu_{k_1,\alpha}$ and
assume that  
\begin{align}
\tilde{k}\not\in S:=\defset{k\in (\T^d)^{\set{2,3,4}}}{k_1+k_2-k_3-k_4=0,\
  \omn(k,\sigma)=\alpha} \, . 
\end{align}
Since $\omega$ is continuous, $S$ is closed, and thus there is $\delta>0$ such
that $B(\tilde{k},2\delta)\subset S^c$.  We can then choose a sequence of
continuous approximants $\phi_n\in C((\T^d)^3)$ of the characteristic function
of the open ball $B(\tilde{k},\delta/2)$ such that the sequence converges
pointwise, has values in $[0,1]$ and is equal to $0$ for any $k'$ with
$|k'-\tilde{k}|\ge \delta$.  By dominated convergence, then
$\nu(B(\tilde{k},\delta/2))= \lim_{n\to\infty} \int\!\nu(\rmd^3 k')
\phi_n(k')$.  Here
\begin{align}\label{eq:phinlim}
  \int \!\nu(\rmd^3 k')\, \phi_n(k')= \lim_{\vep\to 0^+} \int_{(\T^d)^2}
  \frac{\rmd k_2 \rmd k_3}{\pi} \,  \phi_n(k_2,k_3,k_1+k_2-k_3)
  \frac{\vep}{\vep^2+(\omt((k_1,k_2,k_3),\sigma)-\alpha)^2}\, . 
\end{align}
Let $S_1$ denote the support of the function
$(k_2,k_3)\mapsto\phi_n(k_2,k_3,k_1+k_2-k_3)$.  If $S_1$ is empty, then the
map is zero and $\int \!\nu(\rmd^3 k')\, \phi_n(k')=0$.  Assume thus $S_1\ne
\emptyset$. 
If $(k_2,k_3)\in S_1$, then $k':=(k_2,k_3,k_1+k_2-k_3)\in \supp \phi_n$, and
thus $k'\not \in S$ and therefore
$\omt((k_1,k_2,k_3),\sigma)=\omn((k_1,k'),\sigma)\ne \alpha$.  Since $S_1$ is
compact and $(k_2,k_3)\mapsto |\omt((k_1,k_2,k_3),\sigma)-\alpha|$ is
continuous, there is a minimum $c_0$ of the map on $S_1$, and $c_0>0$.  But
then the integrand in (\ref{eq:phinlim}) is bounded by $\vep /c_0^2$, which
implies that $\int \!\nu(\rmd^3 k')\, \phi_n(k')=0$.  Therefore,
$\nu(B(\tilde{k},\delta/2))=0$ and thus  
$\tilde{k}$ is not in the support of $\nu_{k_1,\alpha}$.  This proves item
\ref{it:support}. 

Item \ref{it:fomint} is then a consequence of the fact that, since the support
of $\nu_{k_1,\alpha}$ is part of the Borel set $S$, necessarily for any
continuous function $G$ one has 
$\int \rmd\nu_{k_1,\alpha} G= \int_S \rmd\nu_{k_1,\alpha} G $.   
To prove item \ref{it:contin}, 
consider the map $I:(k_1,\alpha) \mapsto \int\!\nu_{k_1,\alpha,\sigma}(\rmd^3
k')\, G(k_1,k',\alpha)$ for 
some $G\in C((\T^d)^4\times\R)$.  Set $r:=1+\sup_{k}|\omt(k,\sigma)|$, which
is finite since $\omt$ is continuous, and choose a function $f\in C(\R)$ such
that $0\le f\le 1$ and $f(x)=1$ for all $|x|\le r-1$ and $f(x)=0$ for all
$|x|\ge r$.  Define $\bar{G}(k,\alpha):= f(\alpha)G(k,\alpha)$. 
By the Stone-Weierstrass theorem, we can find a sequence $\psi_n(k,\alpha)$
such that it converges uniformly to $\bar{G}$ on $(\T^d)^4\times [-r,r]$ and
each $\psi_n$ is a linear combination of functions of the form $\rme^{-\ci \pi
  x_0\cdot\alpha/r}  \rme^{-\ci 2\pi \sum_{i=1}^4 k_i \cdot x_{i}}$ with
$x_j\in \Z^d$, $j=0,1,\ldots,4$.
Clearly, $I(k_1,\alpha)=\int\!\nu_{k_1,\alpha,\sigma}(\rmd^3 k')\,
\bar{G}(k_1,k',\alpha)$ if $|\alpha|\le r-1$, but this in fact holds for all
$\alpha$ since for $|\alpha|>r-1$ we have
$\int\!\nu_{k_1,\alpha,\sigma}(\rmd^3 k')\,
\bar{G}(k_1,k',\alpha)=0=I(k_1,\alpha)$ by item \ref{it:support}.   
By the previous results, each of the functions $(k_1,\alpha) \mapsto
\int\!\nu_{k_1,\alpha,\sigma}(\rmd^3 k')\, \psi_n(k_1,k',\alpha)$ is
continuous and hence so is $I$ on $(\T^d)^4\times [-r,r]$.  However,
$I(k_1,\alpha)=0$ for all $|\alpha|\ge r$, and thus $I$ is everywhere
continuous.  

To prove item \ref{it:alphaprop}, suppose $F\in C((\T^d)^3\times \R)$ and
$m>\sup_k |\omt(k,\sigma)|$, and fix $k_1\in \T^d$.  
By the previous results, we can then apply dominated convergence and Fubini's
theorem and conclude that  
\begin{align}
 & \int_{-m}^m\!\rmd \alpha\left(\int\!\nu_{k_1,\alpha,\sigma}(\rmd^3 k')\,
   F(k',\alpha)\right) 
\nonumber \\ & \quad
 = \lim_{\vep\to 0^+} \int_{(\T^d)^2} \frac{\rmd k_2 \rmd k_3}{\pi}
 \int_{-m}^m\!\rmd \alpha \,  F(k_2,k_3,k_1+k_2-k_3,\alpha) 
  \frac{\vep}{\vep^2+(\omt((k_1,k_2,k_3),\sigma)-\alpha)^2} \, .
\end{align}
We change the integration variable $\alpha$ to $s=(\omt-\alpha)/\vep$, where
$\omt:=\omt((k_1,k_2,k_3),\sigma)$.  The resulting integral is 
$\int_{(\omt-m)/\vep}^{(\omt+m)/\vep}\!\rmd s \, \frac{1}{1+s^2}
F(k_2,k_3,k_1+k_2-k_3,\omt-\vep s)$ which is uniformly bounded and, since
$|\omt|<m$, it approaches 
$\pi F(k_2,k_3,k_1+k_2-k_3,\omt)$ when $\vep\to 0^+$.  Dominated convergence
can thus be applied to conclude that 
$\int_{-m}^m\!\rmd \alpha\left(\int\!\nu_{k_1,\alpha,\sigma}(\rmd^3 k')\,
  F(k',\alpha)\right)= \int_{(\T^d)^2}\!\rmd k_2 \rmd k_3 \,
F(k_2,k_3,k_1+k_2-k_3,\omt((k_1,k_2,k_3),\sigma))$. This concludes the proof
of the Proposition. 
\end{proof}

In particular, the result thus implies that the (up to now formal)
$\delta$-functions in the definition of $\Ccoll$ correspond to a well-defined
measure.  From now on we denote it by $\nu_{k_1}$, $k_1\in \T^d$, i.e., we set
$\nu_{k_1}:=\nu_{k_1,0,(1,1,-1,-1)}$.  The following results explain how we
use it to define $\Ccoll[W]$ if $W\in \Lphys$.  This somewhat indirect
construction appears necessary since $k_1\mapsto \int\!
\nu_{k_1,0,\sigma}(\rmd^3 k) F(k)$ might not even map sets of Lebesgue measure
zero on $(\T^d)^3$ into sets of measure zero on $\T^d$.  Nevertheless,
Corollary \ref{th:l2colloper} shows that the dissipative part of the collision
operator, which only contains products of $L^\infty(\T^d)$-functions, can be
meaningfully extended into a map from $(L^\infty(\T^d))^3$ to
$L^\infty(\T^d)$. 

We begin with a result which shows that ``Fubini's theorem'' works despite the
$\delta$-functions if the integrand is continuous. 
\begin{corollary}\label{th:swapkcoroll}
Suppose that $\omega:\T^d\to \R$ satisfies (DR\ref{it:DR1}) and
(DR\ref{it:DRdisp}) and suppose $\alpha\in \R$ and $\sigma\in \set{-1,1}^4$
are given. 
Let $\nu_{k,\alpha,\sigma}$ denote measures satisfying
Proposition \ref{th:defC1}, and set
$\sigmapn{2}:=(\sigma_2,\sigma_1,\sigma_3,\sigma_4)$,
$\sigmapn{3}:=(\sigma_3,\sigma_2,\sigma_1,\sigma_4)$ 
and $\sigmapn{4}:=(\sigma_4,\sigma_2,\sigma_1,\sigma_3)$.
If $G:(\T^d)^4\to \C$ is continuous then
\begin{align}\label{eq:Gintiter}
&  \int_{\T^d}\!\rmd k_1 \left( \int\!\nu_{k_1,\alpha,\sigma}(\rmd^3 k')\,
  G(k_1,k'_1,k'_2,k'_3) \right) 
\nonumber \\ & \quad
=  \int_{\T^d}\!\rmd k_2 \left( \int\!\nu_{k_2,\alpha,\sigmapn{2}}(\rmd^3
  k')\, G(k'_1,k_2,k'_2,k'_3)\right) \, . 
\nonumber \\ & \quad
=  \int_{\T^d}\!\rmd k_3 \left( \int\!\nu_{k_3,\alpha,\sigmapn{3}}(\rmd^3
  k')\, G(k'_2,-k'_1,k_3,-k'_3)\right) \, . 
\nonumber \\ & \quad
=  \int_{\T^d}\!\rmd k_4 \left( \int\!\nu_{k_4,\alpha,\sigmapn{4}}(\rmd^3
  k')\, G(k'_2,-k'_1,-k'_3,k_4)\right) \, . 
\end{align}
\end{corollary}
\begin{proof}
Assume $G$ to be continuous. By Proposition \ref{th:defC1}, all four of the
above integrals are over continuous functions and hence well-defined, and we
can apply dominated convergence and Fubini's theorem to conclude that 
\begin{align}\label{eq:beginI4}
& \int_{\T^d}\!\rmd k_4 \left( \int\!\nu_{k_4,\alpha,\sigmapn{4}}(\rmd^3 k')\,
  G(k'_2,-k'_1,-k'_3,k_4)\right)  
\nonumber \\ & \quad
= \lim_{\vep\to 0^+} \int_{(\T^d)^3} \frac{\rmd k_4 \rmd k'_1 \rmd k'_2}{\pi}
G(k'_2,-k'_1,k'_2-k'_1-k_4,k_4) 
  \frac{\vep}{\vep^2+(\omt((k_4,k'_1,k'_2),\sigmapn{4})-\alpha)^2}
\, .
\end{align}
By Fubini's theorem, the value of the integral on the right hand side can be
obtained also by iterating the three integrals in an arbitrary order.  Choose
to do $k_4$ first and change there the integration variable to
$k_3=k'_2-k'_1-k_4$.  Then $k_4=k'_2-k'_1-k_3$ and, since
$\omega(-k)=\omega(k)$, also $\omega(k_4+k'_1-k'_2)=\omega(k_3)$.  Do next
$k'_1$, and change the integration variable to $k_2=-k'_1$.  Rename the last
integration variable to $k_1$.  Then  
$\omt((k_4,k'_1,k'_2),\sigmapn{4})=\omt((k_1,k_2,k_3),\sigma)$, and we can
conclude that the integral is equal to  
\begin{align}
& \int_{\T^d}\! \rmd k_1 \left( \int_{\T^d}\! \rmd k_2 \left( \int_{\T^d}
    \frac{\rmd k_3}{\pi} G(k_1,k_2,k_3,k_1+k_2-k_3) 
  \frac{\vep}{\vep^2+(\omt((k_1,k_2,k_3),\sigma)-\alpha)^2} \right)  \right) 
\, .
\end{align}
Then, by applying Fubini's and dominated convergence theorem, as well as
Proposition \ref{th:defC1}, we find that (\ref{eq:beginI4}) is equal to  
$\int_{\T^d}\!\rmd k_1 \left( \int\!\nu_{k_1,\alpha,\sigma}(\rmd^3 k')\,
  G(k_1,k'_1,k'_2,k'_3) \right)$.    This proves the equality of the first and
last of the expressions in (\ref{eq:Gintiter}).  The proofs that the other two
expressions are equal to the first one are very similar, only simpler, and we
skip them here.

\end{proof}

\begin{lemma}\label{th:l2colloper0}
   Assume that $\omega$ satisfies (DR\ref{it:DR1})  and (DR\ref{it:DRdisp}).  
Consider arbitrary $w_i \in L^\infty(\T^d)$, $i=1,2,3$, and sequences
$w_{i,n}\in C(\T^d)$, $n\in \N$, such that $w_{i,n}\to w_i$ in $L^2$-norm  
and $|w_{i,n}(k)|\le \norm{w_i}_\infty$ for all $i,n,k$.  Then
there is $\mathcal{C}_0\in L^2(\T^d)$ such that the sequence of continuous
functions  
$k_1\mapsto \int_{(\T^d)^3} \nu_{k_1}(\rmd^3 k') \prod_{i=1}^3 w_{i,n}(k'_i)$
converges in $L^2$ to $\mathcal{C}_0$ as $n\to \infty$.    
In addition, there is a constant $C$, which depends only on $\omega$, such
that
\begin{align}\label{eq:C0infbound}
  &\norm{\mathcal{C}_0}_{L^\infty} \le C \prod_{i=1}^3 \norm{w_i}_{L^\infty}\, ,
\end{align}
and, if $w'_i \in L^\infty(\T^d)$, $i=1,2,3$, and $w'_{i,n}\in C(\T^d)$, $n\in
\N$, are another collection of functions as above, and $\mathcal{C}'_0$
denotes the corresponding limit, then 
\begin{align}\label{eq:C0L2cont}
  &\norm{\mathcal{C}_0-\mathcal{C}'_0}_{L^2} \le C m^2 \sum_{i=1}^3
  \norm{w_i-w'_i}_{L^2}\, , 
\end{align}
where $m:=\max_i (\norm{w_i}_\infty,\norm{w'_i}_\infty)$.

Therefore, we can identify $\mathcal{C}_0$ with a unique map
$(L^\infty(\T^d))^3 \to L^\infty(\T^d)$.  This map satisfies all of the
following properties: 
\begin{enumerate}
  \setlength{\itemsep}{0pt}
  \item It is linear in each of the three arguments and commutes with complex conjugation,
    $\mathcal{C}_0[w_1^*,w_2^*,w_3^*] = \mathcal{C}_0[w_1,w_2,w_3]^*$. 
  \item The bounds (\ref{eq:C0infbound}) and (\ref{eq:C0L2cont}) hold for
    $\mathcal{C}_0=\mathcal{C}_0[w_1,w_2,w_3]$ 
    and $\mathcal{C}'_0=\mathcal{C}_0[w'_1,w'_2,w'_3]$.
  \item If  $w_i\in C(\T^d)$ for all $i$, then $\mathcal{C}_0[w_1,w_2,w_3]\in
    C(\T^d)$ and for every $k_1\in \T^d$ 
\begin{align}\label{eq:C0contdef}
  \mathcal{C}_0[w_1,w_2,w_3](k_1) = \int_{(\T^d)^3} \nu_{k_1}(\rmd^3 k')
  \prod_{i=1}^3 w_{i}(k'_i)\, . 
\end{align}  
\end{enumerate}
\end{lemma}
\begin{proof}
Suppose $w_{i,n}$, $w'_{i,n}$ satisfy the assumptions of the Lemma.  Since
they are continuous, setting  
$g_{i,n}(k_1):=\int_{(\T^d)^3} \nu_{k_1}(\rmd^3 k') \prod_{i=1}^3
w_{i,n}(k'_i)$ and  
$g'_{i,n}(k_1):=\int_{(\T^d)^3} \nu_{k_1}(\rmd^3 k') \prod_{i=1}^3
w'_{i,n}(k'_i)$  yields continuous functions on $\T^d$, by Proposition
\ref{th:defC1}. 
Denote $C:= \sup_{k_1,\sigma'} \int \rmd \nu_{k_1,0,\sigma'}$, which is finite
by Proposition \ref{th:defC1}, and set $m:=\max_i
(\norm{w_i}_\infty,\norm{w'_i}_\infty)<\infty$.  Then we have the obvious
bounds 
$|g_{i,n}(k_1)|\le C \prod_{i=1}^3 \norm{w_{i}}_\infty$ and
$|g'_{i,n}(k_1)|\le C \prod_{i=1}^3 \norm{w'_{i}}_\infty$, 
and, by telescoping, we also find that for any $k_1$ 
\begin{align}
  |g_{i,n}(k_1)-g'_{i,n}(k_1)| \le m^2 \sum_{i=1}^3 
 \int_{(\T^d)^3} \nu_{k_1}(\rmd^3 k')   |w_{i,n}(k'_i)-w'_{i,n}(k'_i)|\, .
\end{align}
H\"older's  inequality and $\int\rmd k_1=1$ imply that $\int \rmd k_1
\left(\int \nu_{k_1}(\rmd^3 k')   |w_{i,n}(k'_i)-w'_{i,n}(k'_i)|\right)^2 \le
C \int \rmd k_1 \int \nu_{k_1}(\rmd^3 k')   |w_{i,n}(k'_i)-w'_{i,n}(k'_i)|^2$.
By Corollary \ref{th:swapkcoroll}, the last expression is equal to $C \int
\rmd k_1 |w_{i,n}(k_1)-w'_{i,n}(k_1)|^2 \int \nu_{k_1,0,\sigmapn{i+1}}(\rmd^3
k') \le C^2 \norm{w_{i,n}-w'_{i,n}}^2_{L^2}$.  
Therefore,
\begin{align}\label{eq:deltagbound}
  \norm{g_{i,n}-g'_{i,n}}_{L^2} \le C m^2 \sum_{i=1}^3
  \norm{w_{i,n}-w'_{i,n}}_{L^2} \,. 
\end{align}

Consider then some $n_0\in \N$, and define $w''_{i,n}:= w_{i,n+n_0}$.  Then
$w''_{i,n}\to w_i$ in $L^2$ and    
we can apply the above results to the sequences $w''_{i,n}$.  Since then
$g''_{i,n}=g_{i,n+n_0}$ and $w_{i,n}-w''_{i,n}\to 0$, the bound in
(\ref{eq:deltagbound}) proves that $g_{i,n}$ is a Cauchy sequence in $L^2$.
Hence the $L^2$-limit $\mathcal{C}_0$ exists and there is a subsequence
$(n_\ell)$ such that 
$\mathcal{C}_0(k):= \lim_{\ell\to \infty} g_{i,n_\ell}(k)$ for Lebesgue almost
every $k\in \T^d$.  At every such point we thus have  
$|\mathcal{C}_0(k)|\le C \prod_{i=1}^3 \norm{w_{i}}_\infty$.  Hence
(\ref{eq:C0infbound}) holds, and we have proven the first part of the Lemma. 

These results can also be applied to the sequences $w'_{i,n}$, and the
corresponding limit is given by  
$\mathcal{C}'_0=\lim_{n} g'_{i,n}$.  Thus taking $n\to \infty$ in
(\ref{eq:deltagbound}) proves (\ref{eq:C0L2cont}). 

For the final claim, suppose that $w_i \in L^\infty(\T^d)$, $i=1,2,3$, are
arbitrary.  For each $w_i$, we can choose a representative such that
$|w_i(k)|\le \norm{w_i}_\infty$ for \defem{every} $k$.  Then Lusin's theorem
implies that there are sequences $v_{i,n}\in C(\T^d)$ such that $|v_{i,n}|\le
\norm{w_i}_\infty$ and $v_{i,n}(k)\to w_i(k)$ almost everywhere.  By dominated
convergence, then $v_{i,n}\to w_i$ in $L^2$. 
Thus we can define $\mathcal{C}_0[w]\in L^\infty(\T^d)$, $w=(w_1,w_2,w_3)$, by
using the sequences $(v_{i,n})$ in the above.  Suppose $w_{i,n}$ is some other
sequence as in the Lemma, and let $\mathcal{C}_0$ denote the corresponding
limit.  By (\ref{eq:C0L2cont}), then
$\norm{\mathcal{C}_0-\mathcal{C}_0[w]}_{L^2}=0$, and hence
$\mathcal{C}_0(k)=\mathcal{C}_0[w](k)$ almost everywhere.  Therefore,
$\mathcal{C}_0[w]$ does not depend on the choice of the approximating
sequence.  In particular, then for continuous functions (\ref{eq:C0contdef})
holds, and proving linearity and commutation with conjugation is
straightforward. 
This finishes the proof of the Lemma.
\end{proof}

\begin{corollary}\label{th:l2colloper}
   Assume that $\omega$ satisfies (DR\ref{it:DR1})  and (DR\ref{it:DRdisp})
   and define $\mathcal{C}_0$ as in Lemma \ref{th:l2colloper0}.  For any 
   $W\in\Lphys$ 
   and $i,i'\in \set{1,2}$ we define $\Ccoll[W]_{ii'}\in L^\infty(\T^d)$ by
   using linearity and $\mathcal{C}_0$: we set  
   \begin{align}\label{eq:collop2b}
& \Ccoll [W]_{ii'} = \pi \sum_{j\in \set{\pm 1}^4} \Bigl[ \tilde{W}_{ij_1} 
\left(\delta_{j_2 i'} \mathcal{C}_0[\tilde{W}_{j_3j_4},W_{j_1j_2},W_{j_4j_3}] 
 -\delta_{j_3 j_4}
 \mathcal{C}_0[\tilde{W}_{j_2j_3},W_{j_1j_2},W_{j_4i'}]\right) 
\nonumber \\  & \quad 
+  \tilde{W}_{j_1 i'} 
\left(\delta_{i j_2} \mathcal{C}_0[\tilde{W}_{j_4j_3},W_{j_2j_1},W_{j_3j_4}]
 -\delta_{j_4 j_3} \mathcal{C}_0[\tilde{W}_{j_3j_2},W_{j_2j_1},W_{ij_4}]\right)
\nonumber \\  & \quad
-W_{ij_1} 
\left(\delta_{j_2 i'}
  \mathcal{C}_0[W_{j_3j_4},\tilde{W}_{j_1j_2},\tilde{W}_{j_4j_3}] 
 -\delta_{j_3 j_4}
 \mathcal{C}_0[W_{j_2j_3},\tilde{W}_{j_1j_2},\tilde{W}_{j_4i'}]\right) 
\nonumber \\  & \quad 
-W_{j_1 i'} 
\left(\delta_{i j_2}
  \mathcal{C}_0[W_{j_4j_3},\tilde{W}_{j_2j_1},\tilde{W}_{j_3j_4}] 
 -\delta_{j_4 j_3}
 \mathcal{C}_0[W_{j_3j_2},\tilde{W}_{j_2j_1},\tilde{W}_{ij_4}]\right) 
\Bigr] \, ,
\end{align}
where $\delta_{ab}:=\cf(a=b)$ denotes the Kronecker delta.  Collecting the
components defines a matrix function $\Ccoll[W]:\Lphys\to \Lherm$ such that
$\Ccoll[\tilde{W}]=-\Ccoll[W]$.  If $W\in \Xherm$, then $\Ccoll[W]\in \Xherm$
and (\ref{eq:collop2}) holds pointwise.  In addition, 
there is a constant $C$, which depends only on $\omega$, such that for all
$W,W'\in \Lphys$, 
 \begin{align}
   &\norm{\Ccoll[W]}_\infty\le C\, ,\label{eq:Ccollbound}\\
   &\norm{\Ccoll[W]-\Ccoll[W']}_2\le C\norm{W'-W}_2 \, . \label{eq:Ccdiffbound}
 \end{align}
\end{corollary}

\begin{proof}
  If $M$ is a matrix such that $0\le M\le 1$, then $|M_{ij}|\le 1$ for all
  indices $i,j$.  Since then $0\le 1-M\le 1$, we also have
  $|\tilde{M}_{ij}|\le 1$.  Therefore, if $W\in \Lphys$, 
  then $\tilde{W}\in \Lphys$, and 
  $\norm{W_{ij}}_\infty\le 1$ and $\norm{\tilde{W}_{ij}}_\infty\le 1$ for all
  indices $i,j$. 
  Let $c_0$ denote a constant for which the bounds (\ref{eq:C0infbound}) and
  (\ref{eq:C0L2cont}) in Lemma \ref{th:l2colloper0} hold.
  We can conclude that all of the above $\mathcal{C}_0$-terms are well defined
  and each has an $L^\infty$-norm bounded by $c_0$.  Multiplying this with
  $W_{ij}$ or $\tilde{W}_{ij}$ does not increase the bound, and thus we can
  conclude that $\norm{\Ccoll [W]_{ii'}}_{L^\infty}\le 64 \pi c_0$.  Also, it
  is obvious from the definition that $\Ccoll [W]_{ii'}^* = \Ccoll [W]_{i'i}$
  and hence we have proven that $\Ccoll [W]\in \Lherm$ and it satisfies
  (\ref{eq:Ccollbound}).
  The property $\Ccoll[\tilde{W}]=-\Ccoll[W]$ is also an immediate consequence
  of the definition (\ref{eq:collop2b}). 
  
  If $W\in \Xherm$, then each of its component functions is continuous, and
  then by Lemma \ref{th:l2colloper0} each action of $\mathcal{C}_0$ in
  (\ref{eq:collop2b}) is given by an integral over the same Borel measure
  $\nu_{k_1}(\rmd^3 k')$ and the resulting functions are continuous in $k_1$.
  The sums can be collected inside the integral, and the integrand
  expressed in terms of matrix products.  After some algebra, this proves that
  then (\ref{eq:collop2}) holds for every $k_1$.  Therefore, $\Ccoll [W](k)$
  is everywhere a Hermitian matrix, and we can conclude that also $\Ccoll
  [W]\in \Xherm$. 
  
  To prove (\ref{eq:Ccdiffbound}), assume that  $W,W'\in \Lphys$ are given.  
  Then $\norm{\Ccoll[W]-\Ccoll[W']}^2_2 =
  \sum_{ii'} \norm{\Ccoll[W]_{ii'}-\Ccoll[W']_{ii'}}^2_{L^2}$. 
  By (\ref{eq:C0infbound}) and (\ref{eq:C0L2cont}),
  $\norm{\Ccoll[W]_{ii'}-\Ccoll[W']_{ii'}}_{L^2}$ can be bounded by $64 \pi
  \times 4 c_0 \norm{W'-W}_2$.  Hence $\norm{\Ccoll[W]-\Ccoll[W']}_2 \le 8^3
  \pi c_0 \norm{W'-W}_2$.  This concludes the proof of the corollary. 
\end{proof}

We will later also need another consequence of Proposition \ref{th:defC1}: the
level sets of $\omt$ have then Lebesgue measure zero.  Note that this is not
true in general  without assumption (DR\ref{it:DRdisp}), even for smooth
dispersion relations. Consider for instance $\omega$ which coincides with a
linear map in some neighborhood of $0$: then we have $\omt=0$ in some
sufficiently small ball around zero. 

\begin{corollary}\label{th:omtlevelsets}
  Assume that $\omega:\T^d\to \R$ is continuous and satisfies
  (DR\ref{it:DRdisp}), consider a fixed $\sigma \in \set{-1,1}^4$, and define
  $\omt$ as in (\ref{eq:defomsig}).   Then for any $\alpha\in \R$ the set
  $S_\alpha := \defset{k\in (\T^d)^3}{\omt(k,\sigma)=\alpha}$ is compact and
  $\int_{S_\alpha}\!\rmd^3 k=0$. 
\end{corollary}
\begin{proof}
  Fix $\alpha\in \R$, and denote  $c_0:=\int_{S_\alpha}\!\rmd^3 k$. 
  By continuity of $\omt$, $S_\alpha$ is compact, and hence Borel measurable.
  By Fubini's theorem,  
\begin{align}
& \int_{\T^d}\!\rmd k_1 \left(\int_{(\T^d)^2} \frac{\rmd k_2 \rmd k_3}{\pi} \, 
  \frac{\vep}{\vep^2+(\omt(k,\sigma)-\alpha)^2}\right) 
\nonumber \\ & \quad
= \int_{(\T^d)^3} \frac{\rmd^3 k}{\pi} \, 
  \frac{\vep}{\vep^2+(\omt(k,\sigma)-\alpha)^2}
\ge 
\int_{S_\alpha}  \frac{\rmd^3 k}{\pi} \, 
  \frac{\vep}{\vep^2+(\omt(k,\sigma)-\alpha)^2} = \frac{c_0}{\pi \vep}\, ,
\end{align}
for any $\vep>0$.  However, Proposition \ref{th:defC1} implies that the left
hand side converges to a finite value as $\vep\to 0^+$, which is possible only
if $c_0=0$. 
\end{proof}

\section{The regularized initial value problem}
\label{sec:regIPV}

In this section, we investigate the solutions to the regularized evolution
equation.  These solutions will provide a sequence of approximations used in
the proof of the main theorem.  The goal is to prove that the regularized
problem is well-posed and preserves continuity in $k$; in fact, this can be
proven even without the assumption (DR\ref{it:DRpv}). 

The regularization is defined by choosing an arbitrary $\vep>0$, and setting 
$\Ctot^\vep[W]:= \Ccoll [W]+ \Cdisp^\vep [W]$ where $\Cdisp^\vep [W]:=-\ci
[\Heff^\vep[W],W]$ 
and $\Heff^\vep[W] (k_1)$ is defined using the integral in
(\ref{eq:defHeffeps}). 
By Corollary \ref{th:l2colloper}, for a given $W\in \Xherm$ also $\Ccoll
[W]\in \Xherm$ and it satisfies (\ref{eq:collop2}).  It is also
straightforward to check that $\Heff ^\vep[W]\in \Xherm$ for any $W\in
\Xherm$, and hence also $\Cdisp ^\vep[W]\in \Xherm$.  Therefore, the
regularized collision operator is a well-defined map from $\Xherm$ to itself
and we can hope to solve the regularized evolution problem in the space
$\Xherm$.  In fact, we can show that not only is the regularized problem with
fermionic initial data well-posed, but it also preserves the Fermi property
and the conservation laws. 
\begin{theorem}\label{th:main2}
  Suppose $\omega$ satisfies (DR\ref{it:DR1}) and (DR\ref{it:DRdisp}).
  If $W_0\in \Xphys$ and $\vep>0$, then there is a unique $W\in
  C^{(1)}([0,\infty),\Xphys)$ such that $W(0,k)=W_0(k)$ for every $k$, and  
for all $t>0$ and every $k\in \T^d$, 
\begin{align}
  \partial_t W_t(k) = \Ccoll [W_t](k) - \ci [\Heff^\vep [W_t](k),W_t(k)]\, ,
\end{align}
where $W_t(k):=W(t,k)$. 
In addition, $W_t$ depends continuously on $W_0$ on any compact interval of
$[0,\infty)$, and the solution conserves total energy and spin: equalities
(\ref{eq:claw1}) and (\ref{eq:claw2}) hold for all $t\ge 0$. 

Moreover, to every $k\in \T^d$ and $0\le s\le t$ we can then attach a unitary
matrix $U^\vep_{t,s}(k;W_0)$ such that $U^\vep_{t,t}(k;W_0)=1$ and for which  
the map $s\mapsto U^\vep_{t,s}(k;W_0)$ belongs to 
$C^{(1)}([0,t],\C^{2\times 2})$ with 
  \begin{align}
  & \partial_s U^\vep_{t,s}(k;W_0)=\ci U^\vep_{t,s}(k;W_0) 
\Heff^\vep[W_s](k)\, .
  \end{align}
Then also
\begin{align}
  & W_t(k) = U^\vep_{t,0}(k;W_0) W_0(k) U^\vep_{t,0}(k;W_0)^* 
%\nonumber \\ & \quad 
   + \int_0^t\!\rmd s\, U^\vep_{t,s}(k;W_0) \Ccoll[W_s](k)
   U^\vep_{t,s}(k;W_0)^*\, . 
\end{align}
\end{theorem}

\begin{proof}
Suppose $\omega$ satisfies (DR\ref{it:DR1}) and (DR\ref{it:DRdisp}), and
consider a fixed $\vep>0$.  
As explained above, then for any $W\in \Xherm$, both $\Ccoll[W]$ and $\Heff
^\vep[W]$ are defined directly as integrals over the appropriate measures, and
we can use matrix-algebraic manipulations in the integrands to simplify the
formulae.  First, we observe that in both cases the highest order monomial
terms cancel out inside the integrand.  For $\Heff^\vep$, we then find the
following alternative, slightly less symmetric but shorter, expression 
\begin{align}\label{eq:defHeff2}
& \Heff ^\vep[W](k_1) = -\frac{1}{2} \Ckint \delta(\underline{k})\,
\frac{\underline{\omega}}{\underline{\omega}^2 + \vep^2} \, 
\nonumber \\  & \quad \times
\left(J[W_4-W_2] W_3 + W_3 J[W_4-W_2] + 
J[W_2 \tilde{W}_4 + \tilde{W}_4 W_2 ]\right)  \, .
\end{align}
Furthermore, we can now split the dissipative part into sum of a ``gain'' and
a ``loss'' term.  Defining 
\begin{align} \label{eq:defCgain}
 & \Cgain [W](k_1):=\pi \Ckint \delta(\underline{k}) \delta(\underline{\omega})
 % \nonumber \\ & \quad \times
 \left( W_3 J[\tilde{W}_2 W_4]
+ J[ W_4 \tilde{W}_2] W_3 \right) \, ,\\ \label{eq:defCloss}
 & \Closs [W](k_1):=\pi \Ckint  
\delta(\underline{k}) \delta(\underline{\omega})
 %  \nonumber \\ & \quad \times
 \left( J[ W_4 \tilde{W}_2]W_3 
+ J[\tilde{W}_4  W_2] \tilde{W}_3 \right) \, ,
\end{align}
allows to rewrite (\ref{eq:collop2}) as
\begin{align}
  \Ccoll[W](k_1) = \Cgain [W](k_1) - \Closs [W](k_1) W(k_1)- W(k_1) \Closs
  [W](k_1)^*\, . 
\end{align}
Here the first term is called the \defem{gain term} and the rest, the
\defem{loss term}. 

As mentioned in the beginning of this section, under the present assumptions,
$\Ctot ^\vep$ maps $\Xherm$ into itself.  There is also an additional
symmetry: if $W$ is a solution to  
$\partial_t W_t = \Ctot^\vep[W_t]$ with initial data $W_0$, then $\tilde{W}$
is also a solution, with initial data $\tilde{W}_0$.  To see this, note that
then $\partial_t \tilde{W}_t=-\Ctot^\vep[W_t]$ where $-\Ctot^\vep[W_t]=\Ctot
^\vep[\tilde{W}_t]$, as can be seen by using the property
$(\tilde{W})\tilde{\,}=W$ in the integral representations (\ref{eq:collop2})
and in (\ref{eq:defHeffeps}):  this shows that 
$\Ccoll [\tilde{W}] = - \Ccoll [W]$ and 
$\Heff ^\vep[\tilde{W}]=\Heff ^\vep[W]$, and hence also that 
$\Cdisp^\vep [\tilde{W}] = - \Cdisp^\vep [W]$.

In the proof of the Proposition we follow the strategy used by Dolbeault in
\cite{dolb94}, albeit with a somewhat  different truncation procedure.  We
begin by introducing a truncation to the collision operator which will be
employed to ensure the existence of global solutions.  For this,  we first set 
\begin{align}
  \Phi(x) := \begin{cases}
                 0, & \text{if } x< 0\, ,\\
                 x, & \text{if } 0\le x\le 1\, ,\\
                 1, & \text{if } x> 1\, .
                 \end{cases}
\end{align}
For any Hermitian matrix $M$, we then define a positive matrix $\trunc{M}$ via
symbolic calculus. 
Explicitly, let $\lambda_i\in \sigma(M)$, with $i$ counting the eigenvalues,
and denote by $P_i$ the corresponding spectral projection operators (which
have a two-dimensional range, in case of a degenerate eigenvalue).  Then $M=
\sum_{i} \lambda_i P_i$  
is the spectral decomposition of $M$ and we define
\begin{align}\label{eq:defPhiM}
\Phi[M]:=\sum_{i}\Phi(\lambda_i)P_i\, .
\end{align}
This results in a Lipschitz map, as the following Lemma shows.
\begin{lemma}\label{th:basicPhi}
  There is a constant $C>0$ such that $\norm{\trunc{M'}-\trunc{M}}\le C
  \norm{M'-M}$ for all $M,M'\in \C^{2\times 2}$ which are Hermitian.  In
  addition, then $1-\trunc{M}= \trunc{1-M}$, $0\le \trunc{M}\le 1$ and
  $\trunc{M}=M$ if $0\le M\le 1$. 
\end{lemma}
\begin{proof}
We first observe that $\Phi$ satisfies
\begin{align}\label{eq:Phiabs2}
\Phi(x) = \frac{1}{2}(1 + \left | x \right |- \left | 1-x \right |)\, ,\quad
x\in \R\, .  
\end{align}
Let $|M|$ denote the absolute value of a matrix $M$, defined by 
$\left |M \right | := (M^\ast M)^ \frac{1}{2}$. 
Using the above spectral decomposition of $M$, we find that $|M|= \sum_{i}
|\lambda_i| P_i$, and hence by (\ref{eq:Phiabs2})  
\begin{align}\label{eq:Mtoabs}
\trunc{M}= \frac{1}{2}(1  + \left | M \right |- \left | 1 -M \right |) \, .
\end{align}
The map $M \rightarrow \left | M \right |$ is Lipschitz continuous in the
Hilbert-Schmidt norm (this is proven, for instance, in 
\cite[Theorem 1]{AY81}).  
Hence (\ref{eq:Mtoabs}) implies that $M\mapsto \trunc{M}$ is also
Lipschitz continuous, and there is a pure constant $C$ such that
$\norm{\trunc{M'}-\trunc{M}}\le C \norm{M'-M}$ for all $M,M'$. 

By (\ref{eq:Phiabs2}), clearly $\Phi(1-x)=1-\Phi(x)$ for all $x\in \R$.  Since
we have $\sum_i P_i = 1$ 
in the spectral decomposition, it follows that $1-\trunc{M}= \trunc{1-M}$. 
The other properties of $\trunc{M}$ listed in the Lemma are obvious
consequences of the definition (\ref{eq:defPhiM}). 
\end{proof}

As the first step in the proof, we will show that for any $w_0\in \Xherm$
there is a unique global solution $w\in C^{(1)}([0,\infty),\Xherm)$ which
solves the partially truncated evolution equation 
\begin{align}\label{eq:defCtr}
 \partial_t w_t(k) = 
  \Ctr [w_t](k)\, ,\qquad \Ctr[w]:= \Ccolltr [w]+ \Cdisptr [w]\, ,
\end{align}
where, defining $\Phi[w](k):= \Phi[w(k)]$ for $k\in \T^d$, $w\in \Xherm$,
\begin{align}\label{eq:defCcolltr}
 & \Ccolltr [w](k) := \Cgain [\trunc{w}](k) - 
 \Closs [\trunc{w}](k) w(k)- w(k) \Closs [\trunc{w}](k)^*\, , 
 \\ \label{eq:defCdisptr}
 & \Cdisptr [w](k) :=  -\ci [\Heff^\vep [\trunc{w}](k), w(k)]\, .
\end{align}
To prove this, we rely on the standard fixed point methods in the Banach
spaces $Y_{t_0}:=C([0,t_0],\Xherm)$, $t_0>0$, equipped with the norms
$\norm{w}_Y := \sup_{t,k} |w_t(k)|=\sup_t \norm{w_t}$.   
Given $W_0\in \Xherm$ we define for all $w\in Y_{t_0}$, $k\in \T^d$, 
$0\le t\le t_0$, 
\begin{align}\label{eq:defTw}
    \mathbf{T}[w]_t(k) := W_0(k) + \int_0^{t}\!\rmd s\, \Ctr [w_s](k)\, .
\end{align}
We begin by proving that for any $W_0\in \Xherm$, there is a non-increasing
function  
$\theta_0(\norm{W_0})>0$ such that $\mathbf{T}$ is a contractive mapping on
the closed ball $\overline{B}(W_0,1)$ of $Y_{\theta_0}$, where with a slight
abuse of notation we have denoted by $W_0$ also its time-constant extension,
i.e., the function $F \in Y_{\theta_0}$ for which $F_t(k):= W_0(k)$ for all
$t,k$. 

Suppose $w\in \Xherm$ and $k\in \T^d$ are arbitrary.  Since $\norm{J[M]}\le
\norm{M}$ for any $M\in \C^{2\times 2}$, we have $\norm{w_1 J[w_2 w_3]}\le
\prod_{i=1}^3 \norm{w_i}$ for any choice of $w_i\in \C^{2\times 2}$,
$i=1,2,3$.  Together with Proposition \ref{th:defC1}, this implies that there
is a constant $C_0$, depending only on $\omega$, such that  
\begin{align} 
 \|\Cgain [w](k)\|,\ \|\Closs [w](k)\| \le C_0 (1+\norm{w})^3 \, .
\end{align}
Therefore, $\|\Ccolltr [w](k)\|\leq C_0 (1+\norm{\Phi[w]})^3 
(1 + 2 \norm{w(k)})$,
and since $\norm{\Phi[w]}\le 2$, we can find a constant $C_1$, also depending
only on $\omega$, such that  
\begin{align}
\|\Ccolltr [w](k)\|\leq C_1 (1 + \norm{w(k)})\, .
\end{align}
Similarly, we obtain that there is a constant $C^\vep_2$, depending only on
$\vep$ and $\omega$, such that $\norm{\Heff^\vep [\Phi[w]](k)} \le
C^\vep_2/2$, and hence $\norm{\Cdisptr [w](k)} \le C^\vep_2\norm{w(k)}$. 
Thus
\begin{align}
\norm{\Ctr [w](k)} \le R (1 + \norm{w(k)})\, ,
\end{align}
where $R:=C_1+ C^\vep_2$ depends only on $\vep$ and $\omega$. 

Suppose $w\in \Xherm$. 
By item \ref{it:contin} of Proposition \ref{th:defC1}, 
then both $\Cgain [w]$ and $\Closs [w]$ are continuous functions in $k$.
On the other hand, for instance by using dominated convergence, we find that
also $k\mapsto  
\Heff^\vep [w](k)$ is continuous.  
By Lemma \ref{th:basicPhi}, the map $M\mapsto \Phi[M]$ is continuous in the
matrix norm, 
and thus we can conclude that for any $w\in \Xherm$ we have 
$\Ctr [w]\in \Xherm$ with  
\begin{align}\label{eq:bd}
\norm{\Ctr [w]} \le R (1 + \norm{w})\, .
\end{align}
Note that $\Cgain [\Phi[w]](k)$ and $\Heff^\vep [\Phi[w]](k)$ are clearly
Hermitian for $w\in \Xherm$, $k\in \T^d$, and thus so is $\Ctr [w](k)$. 

Let us next consider the continuity properties of the map $w\mapsto \Ctr [w]$
in $\Xherm$. 
The following Lemma generalizing a result of Seiler and Simon \cite{SeSi75}
will become useful for this purpose.  The corresponding statement and a proof
for scalar valued functions is given in \cite[Theorem 4.1]{gohberg00}.  In
fact, the proof carries over verbatim for matrix valued functions, as soon as
one understand all integrals as ``vector valued'' in the sense of used in
topological vector spaces \cite{rudin:fa}.  The full proof is included here
mainly for the sake of completeness. 
\begin{lemma}[Lipschitz bounds]\label{th:lip}
Let $f$ be a complex matrix valued function defined on complex normed linear
space $\mathcal{N}$. Suppose that 
\begin{enumerate}
\item the function $\lambda\mapsto f(A+\lambda B)$ is an entire matrix
  function for all ${A}$, ${B}$ in $\mathcal{N}$, and  
\item there is a monotone non-decreasing function ${g}$ on $[0,\infty)$ such
  that for all $A\in \mathcal{N}$, 
\begin{align}
\left \| f(A) \right \| \leq g(\left \|A  \right \|_{\mathcal{N}}) \, .
\end{align}
\end{enumerate}
Then for all ${A}$, ${B}$ in $\mathcal{N}$,
\begin{align}\label{eq:fABbound0}
\left \| f(A)-f(B) \right \| \leq \left \|A-B  \right \|_{\mathcal{N}} g(\left
  \|A  \right \|_{\mathcal{N}}+\left \|B  \right \|_{\mathcal{N}}+1) \, .   
\end{align}
\end{lemma}

\begin{proof}
Let $h(\lambda):=f(\frac{1}{2}(A+B)+\lambda(A-B))$. By assumption, then $h$ is
an entire matrix function and  
\begin{align}
\left \| f(A)-f(B) \right \|=\left \|
  h\left(\frac{1}{2}\right)-h\left(-\frac{1}{2}\right) \right \|\leq
\sup_{-\frac{1}{2}\leq t\leq \frac{1}{2}} \left \|h^\prime (t)  \right \| \,
. 
\end{align}
By assumption, the matrix valued map $h$ is holomorphic on $\C$.   
From the Cauchy contour integral formula (see for instance the proof of
Theorem 3.31 in \cite{rudin:fa} for details) we get for $\rho > 0$ and
$-\frac{1}{2}\leq t\leq \frac{1}{2}$ 
\begin{align}
& \left \|h^\prime (t)\right \|= \frac{1}{2\pi}\left
  \|\,\smash{\oint\limits_{\left | s \right |= \rho}} \frac{h(s+t)}{s^2}
  \mathrm{d}s \right \|\leq\frac{1}{\rho}\sup_{\left | s \right |= \rho} \left
  \| h(s+t) \right \|\leq \frac{1}{\rho}\sup_{\left | \lambda \right |\leq
  \rho + \frac{1}{2}} \left \| h(\lambda) \right \|\, . 
\end{align}
Therefore, for any  $\rho > 0$,
\begin{align}\label{eq:fABbound1}
\left \| f(A)-f(B) \right \|\leq \frac{1}{\rho}\sup_{\left | \lambda 
\right|\leq \rho + \frac{1}{2}} \left \| h(\lambda) \right \|\, .   
\end{align}

If $A=B$, then (\ref{eq:fABbound0}) obviously holds.  Assume thus $A\ne B$,
and set  
$\rho := \left \|A-B  \right \|^{-1} _{\mathcal{N}}$.   Then for any $\left
  |\lambda  \right |\leq \rho + \frac{1}{2}$ we have 
\begin{align}
&\left \|\frac{1}{2}(A+B)+\lambda(A-B) \right \|_{\mathcal{N}}\leq
\frac{1}{2}(\left \|A+B \right \|_{\mathcal{N}} + \left \|A-B \right
\|_{\mathcal{N}}) + \rho \left \|A-B  \right \|_{\mathcal{N}} 
 \nonumber \\  & \quad
\leq \left \|A\right \|_{\mathcal{N}} + \left \|B\right \|_{\mathcal{N}} + 1 \, . 
\end{align}
Hence, if $\left |\lambda  \right |\leq \rho + \frac{1}{2}$, then by the
second assumption  we can estimate 
\begin{align}
\left \|h(\lambda)\right \|= \left \|
  f\!\left(\frac{1}{2}(A+B)+\lambda(A-B)\right) \right \| \leq g\!\left(\left
    \|\frac{1}{2}(A+B)+\lambda(A-B) \right \|_{\mathcal{N}}\right) \, , 
\end{align}
where, by the monotonicity of $g$, the right hand side is bounded by $g(\left
  \|A\right \|_{\mathcal{N}} + \left \|B\right \|_{\mathcal{N}} + 1)$.
Therefore, (\ref{eq:fABbound1}) implies now that also for $A\ne B$  
(\ref{eq:fABbound0}) holds. 
\end{proof}

Consider then some $k\in \T^d$.  Since $\norm{\Cgain[w](k)} \le C_0
(1+\norm{w})^3$, the map $w\mapsto \Cgain[w](k)$ satisfies the
conditions\footnote{In principle, we have defined the map only in the real
  Banach space $\Xherm$.  However, it is obvious that the defining integral
  can also be applied in the complex Banach space $C(\T^d,\C^{2\times 2})$ and
  that none of the bounds used the fact that $w(k)$ is Hermitian.}  of the
Lemma with $g(x):= C_0(1+x)^3$.  Hence, for all $w,w'\in \Xherm$ and any 
$k\in\T^d$ we have  
$\norm{\Cgain[w'](k)-\Cgain[w](k)}\le \norm{w'-w} 
C_0 (\norm{w'}+\norm{w}+ 2)^3$.  
Therefore, by Lemma \ref{th:basicPhi}, there is a
constant $c'_0$, which depends only on $\omega$, such that  
$\norm{\Cgain[\Phi[w']]-\Cgain[\Phi[w]]}\le c'_0 \norm{w'-w}$.  Thus $w
\mapsto \Cgain[\Phi[w]]$ is a Lipschitz map on $\Xherm$.  Completely analogous
reasoning shows that also $w \mapsto \Closs[\Phi[w]]$ 
and $w \mapsto \Heff^\vep[\Phi[w]]$  are Lipschitz maps $\Xherm\to
C(\T^d,\C^{2\times2})$, with Lipschitz constants which depend only on $\omega$
and $\vep$.  Therefore, using also the earlier derived uniform bounds, we can
conclude that there is a constant $R'$, which depends only on $\omega$ and
$\vep$, such that 
\begin{align}\label{eq:bd2}
\norm{\Ctr [w']-\Ctr [w]} \le R' (1 + \norm{w}+\norm{w'}) \norm{w'-w}\,
,\qquad w,w'\in \Xherm\, . 
\end{align}

Consider then an arbitrary $t_0>0$.  If $w\in Y_{t_0}$, dominated
convergence and the above bounds imply that the function $\mathbf{T}[w]$
defined by (\ref{eq:defTw}) is continuous, both in $t$ and in $k$,  and
$\mathbf{T}[w]_t(k)$ is always Hermitian and satisfies the bound 
\begin{align}
\|\mathbf{T}[w]_t(k)- W_0(k)\|\leq R\int_0 ^{t}\!\rmd s \, (1+\norm{w_s}) \, ,
\quad w\in Y_{t_0}\,. 
\end{align}
In addition,  by (\ref{eq:bd2}) for any $w', w\in Y_{t_0}$,
\begin{align}
  \|\mathbf{T}[w']_t(k)-\mathbf{T}[w]_t(k)\|\le t R' (1 +
  \norm{w}_Y+\norm{w'}_Y) \norm{w'-w}_Y \, , 
  \quad t\in [0,t_0],\ k\in \T^d\, .
\end{align}
Set $\theta_0(\norm{W_0}):= c_2/(1+\norm{W_0})$ where $c_2:=
\min((8R')^{-1},(2 R)^{-1})$ depends only on $\omega$ and $\vep$, and thus
$\theta_0$ is a non-increasing function of $\norm{W_0}$ with $\theta_0>0$ for
all $\norm{W_0}<\infty$.   Then, if $t_0\le \theta_0$, the above estimates
imply  
$\|\mathbf{T}[w]- W_0\|_Y\leq 1$ and $
\|\mathbf{T}[w']_t(k)-\mathbf{T}[w]_t(k)\|_Y\le \norm{w'-w}_Y/2$ whenever
$w',w\in Y_{t_0}$ satisfy $\norm{w'-W_0}_Y, \norm{w-W_0}_Y\le 1$. 
This proves the earlier claim that for any $0<t_0\le \theta_0$ the map
$\mathbf{T}$ is a contraction from the closed ball $\overline{B}(W_0,1)$ into
itself.  Thus 
we can conclude from the Banach fixed point theorem that for any $W_0\in
\Xherm$ there is a unique $w\in \overline{B}(W_0,1)\subset Y_{\theta_0}$,
$\theta_0=\theta_0(\norm{W_0})>0$, for which  
\begin{align}\label{eq:fpeqn}
    w_t(k) = W_0(k) + \int_0^{t}\!\rmd s\, \Ctr [w_s](k)\, , \quad 0\le t \le
    \theta_0,\ k\in \T^d\, . 
\end{align}

Fix then $W_0\in \Xherm$ and consider the collection of solutions to
(\ref{eq:fpeqn}), i.e., the set of  
$(t_0,w)$, with $t_0>0$, $w\in Y_{t_0}$, such that (\ref{eq:fpeqn}) holds for
$0\le t\le t_0$.  By the above result we know that this set is not empty, and
extension of functions clearly defines a partial order on it.  By Hausdorff's
maximality principle, there is a maximal totally ordered subset.  Let $T_0$
denote the supremum of the $t_0$ in this set, and define $w_t(k)$ for $0\le
t<T_0$, $k\in \T^d$, by choosing a function from the set with $t_0>t$ (such
a $t_0$ must exist and the value of $w_t(k)$ does not depend on the choice
since the set is totally ordered).  Then $w\in C([0,T_0),\Xherm)$, and it is a
maximal solution to (\ref{eq:fpeqn}) for $0\le t<T_0$.   Now  (\ref{eq:bd})
and  (\ref{eq:fpeqn}) imply 
\begin{align}%\label{eq:fpeqn}
    \norm{w_t(k)} \le \norm{W_0} +  \int_0^{t}\!\rmd s\, R (1+\norm{w_s})\, ,
    \quad 0\le t <T_0\, ,\ k\in \T^d\, . 
\end{align}
The map $s\mapsto \norm{w_s}$ is continuous, and thus Gr\"onwall's lemma can
be applied on the time-interval $[0,t]$.  This allows to conclude that 
\begin{align}\label{eq:wtnormbound}
    \norm{w_t} \le (\norm{W_0} + R t) \rme^{R t}\, , \quad 0\le t <T_0\, .
\end{align}
If we suppose that $T_0<\infty$, this would imply that $\norm{w_t} \le
c:=(\norm{W_0} + R T_0) \rme^{R T_0}$ for all $0\le t <T_0$.  However, if we
then apply the fixed point result to initial data $w_{T_0-\theta_0(c)/2}$ we
obtain an extension up to times $T_0+\theta_0(c)/2$.  This contradicts the
maximality of $w$, and hence necessarily $T_0=\infty$. 

We have now proven that for any $W_0\in \Xherm$ there is $w\in
C([0,\infty),\Xherm)$ which satisfies (\ref{eq:fpeqn}) for all $t,k$.  Since
$\Ctr [w_t](k)$ is continuous, this directly implies that  
$w\in C^{(1)}([0,\infty),\Xherm)$ with $\partial_t w_t = \Ctr [w_t]$, for all
$t\ge 0$, and that $w_t\to W_0$ as $t\to 0^+$.  Thus it provides a global
solution to the truncated Cauchy problem.   
Suppose $v$ is another global solution corresponding to some initial data
$V_0\in \Xherm$. Then it also satisfies (\ref{eq:fpeqn}) for all $t,k$ and
hence also (\ref{eq:wtnormbound}). Therefore, by (\ref{eq:bd2}) 
\begin{align}
 & \norm{v_t(k)-w_t(k)}\le \norm{V_0-W_0}+ \int_0^{t}\!\rmd s\, 
  \norm{\Ctr [v_s]-\Ctr [w_s]}
 \nonumber \\  & \quad
  \le \norm{V_0-W_0}+
  R'  (1+\norm{V_0} +\norm{W_0} + 2 R t) 
  \rme^{R t} \int_0^{t}\!\rmd s\,\norm{v_s-w_s} \, ,
\end{align}
for all $t\ge 0$ and $k\in \T^d$.  Thus Gr\"onwall's lemma implies that, if
$t\in [0,t_0]$, with $t_0>0$ arbitrary, then  
\begin{align}\label{eq:wtrcontdep}
\left \| v_t-w_t \right \|\le  \norm{V_0-W_0} \rme^{t_0 R'  (1+\norm{V_0}
  +\norm{W_0} + 2 R t_0)}\, . 
\end{align}
This shows that the global solution is unique, and also proves 
that it depends continuously on the initial data.

We have now proven that the truncated problem (\ref{eq:defCtr}) is well-posed
for any initial data $W_0\in \Xherm$.  Our next goal is to show that such
solutions preserve the Fermi property, i.e., to show that if $W_0\in \Xphys$, 
then the corresponding solution $w$ satisfies $w_t\in \Xphys$ for all $t\ge
0$.  Suppose this is the case.  Then $\Phi[w_t]=w_t$, and thus
$\Ctr[w_t]=\Ctot^\vep[w_t]$ for all $t$.  It follows that then $w\in
C^{(1)}([0,\infty),\Xphys)$, $w_0=W_0$, and $\partial_t w_t= \Ctot^\vep [w_t]$  
for all $t>0$; therefore, choosing $W=w$ yields a solution satisfying the
conditions of the Theorem.  It is also the only such function: for any $W$ as
in the Theorem, we have 
$W_t\in \Xphys$ implying $\Ctot^\vep[W_t]=\Ctr[W_t]$, and thus $W$ is then a
solution to the truncated problem, hence equal to $w$.  Finally,
(\ref{eq:wtrcontdep}) then immediately implies that the unique solution
depends continuously on the initial data. 

Therefore, to complete the proof of the first part of the Theorem we only need
to show that the above solutions preserve the Fermi property and the
conservation laws.  In fact, for the Fermi property it suffices to show that
the solutions preserve positivity in the matrix sense.   
Indeed, suppose that we have proven that for every $W_0\in \Xherm$ with
$W_0\ge 0$ necessarily $w_t\ge 0$ for all $t$.   We will soon show that also
the truncated problem preserves the $W\to \tilde{W}$ symmetry, i.e.,  we will
show that if $w$ is a solution to $\partial_t w_t = \Ctr[w_t]$ with initial
data $W_0$, then $\tilde{w}$ is a solution with initial data $\tilde{W}_0$.
If $W_0\in \Xphys$, we have $W_0\ge 0$ and $\tilde{W}_0\ge 0$, and as
positivity of solutions is preserved, we can then conclude that $w_t\ge 0$
and $\tilde{w}_t\ge 0$, and thus $0\le w_t\le 1$. 

To prove the symmetry statement suppose $W_0\in \Xherm$ and let $w$ denote the
corresponding solution. 
Then $\partial_t \tilde{w}_t = -\Ctr[w_t]$, and since
$\tilde{w}_0=\tilde{W}_0$ it suffices to show that  
$\Ctr[\tilde{v}]=-\Ctr[v]$ for all $v\in \Xherm$.  For this, first note that 
by Lemma \ref{th:basicPhi}, we have always
$\trunc{\tilde{v}}=(\trunc{v})\tilde{\,}$.  Hence, by the earlier discussion
$\Ctot ^\vep[\trunc{\tilde{v}}]=-\Ctot ^\vep[\trunc{v}]$ and $\Heff
^\vep[\trunc{\tilde{v}}]=\Heff ^\vep[\trunc{v}]$.  Therefore, now $\Cdisptr
[\tilde{v}]=-\Cdisptr [v]$ and also $\Closs [\trunc{\tilde{v}}]= \Closs
[\trunc{v}]$, since it is obvious from the definition (\ref{eq:defCloss}) that
$\Closs [\tilde{W}]= \Closs [W]$.   
Employing the above equalities shows that
\begin{align}\label{eq:ctrtt}
 & \Ctr [\tilde{v}]+\Ctr [v]
 = \Cgain [1-\trunc{v}] + 
  \Cgain [\trunc{v}] - \Closs [\trunc{v}] - \Closs [\trunc{v}]^*\, .
\end{align}
However, for any $W\in \Xherm$, we find directly from the definitions
(\ref{eq:defCgain}) and (\ref{eq:defCloss}) that 
\begin{align} 
 & \Cgain [\tilde{W}](k_1) + \Cgain [W](k_1)  
  \nonumber \\ & \quad
   = \pi \Ckint \delta(\underline{k}) \delta(\underline{\omega})
  \nonumber \\ & \qquad \times
 \left( \tilde{W}_3 J[W_2 \tilde{W}_4]
+ J[ \tilde{W}_4 W_2] \tilde{W}_3 + W_3 J[\tilde{W}_2 W_4]
+ J[ W_4 \tilde{W}_2] W_3
\right)
  \nonumber \\ & \quad =
  \Closs [W](k_1) + \Closs [W](k_1) ^*
\, .
\end{align}
Thus (\ref{eq:ctrtt}) implies $\Ctr [\tilde{v}]=-\Ctr [v]$, and proves the
stated preservation of the $W\to \tilde{W}$ symmetry. 

Therefore, to prove the preservation of the Fermi property, we now only need
to show that the solutions preserve positivity.  The key ingredient in this
proof is Lemma \ref{th:posgain}, which implies that the truncated gain term is
always a nonnegative matrix.  Its proof will rely on the following matrix inequality. 
\begin{lemma}\label{th:matrixpos}
For any $n \times n$ matrices $A, B, C\geq 0$, $n\ge 1$, we have
\begin{align}\label{eq:MI}
 AJ[BC]+CJ[BA]\geq 0 \qand J[AB]C+J[CB]A\ge 0 \, .
\end{align}
\end{lemma}
\begin{proof}
By expanding the definition of $J$, we find that
$AJ[BC]+CJ[BA]=A\trace(BC) + C\trace(AB) - ABC - CBA$.  To prove that this is
nonnegative, choose a complete set of eigenvectors for each Hermitian matrix:
let $(a_i,\alpha_i)_{i=1,\ldots,n}$ be an eigensystem for $A$,    
$(b_i,\beta_i)_{i}$ for $B$, and $(c_i,\gamma_i)_{i}$ for $C$.
Then for any $\psi\in \C^n$, by selecting suitable bases to express the matrix
products and to compute the traces, we obtain 
\begin{align}
 & \spr{\psi}{( A\trace(BC) + C\trace(AB) - ABC - CBA)\psi}
   \nonumber \\ & \quad
 = \sum_{i,j,k=1,\ldots,n} a_i b_j c_k
 \bigl(\spr{\psi}{\alpha_i}\spr{\alpha_i}{\psi}\spr{\beta_j}{\gamma_k}
  \spr{\gamma_k}{\beta_j} 
   +\spr{\psi}{\gamma_k}\spr{\gamma_k}{\psi}\spr{\beta_j}{\alpha_i}
  \spr{\alpha_i}{\beta_j}
   \nonumber \\ & \quad\qquad
 -\spr{\psi}{\alpha_i}\spr{\alpha_i}{\beta_j}\spr{\beta_j}{\gamma_k}
  \spr{\gamma_k}{\psi}
   -\spr{\psi}{\gamma_k}\spr{\gamma_k}{\beta_j}\spr{\beta_j}{\alpha_i}
  \spr{\alpha_i}{\psi}
   \bigr)
      \nonumber \\ & \quad
 = \sum_{i,j,k} a_i b_j c_k
 \bigl| \spr{\psi}{\alpha_i}\spr{\beta_j}{\gamma_k}-\spr{\psi}{\gamma_k}
  \spr{\beta_j}{\alpha_i}
   \bigr|^2 \ge 0\, ,
\end{align}
since, by assumption, $a_i, b_i, c_i\ge 0$.
This implies that the matrix is non-negative.  Taking an adjoint proves then
that also $J[AB]C+J[CB]A\ge 0$. 
\end{proof}
\begin{lemma}\label{th:posgain}
  If $W\in \Xphys$, then $\Cgain[W](k)\ge 0$ for all $k$.  Therefore,
  $\Cgain[\trunc{W}](k)\ge 0$ for all $k$ and $W\in \Xherm$.  
\end{lemma}
\begin{proof}
It follows directly from its definition in Proposition \ref{th:defC1} that the
measure $\nu_{k_1}$ is invariant under the exchange $k_4\leftrightarrow k_3$.
To see this, one can  change the integration variable $k_3$ to
$k_4:=k_1+k_2-k_3$ in (\ref{eq:numaindef}) and note that
$\omt((k_1,k_2,k_1+k_2-k_4),(1,1,-1,-1))=\omt((k_1,k_2,k_4),(1,1,-1,-1))$.
Using this symmetry in the definition (\ref{eq:defCgain}) shows that  
\begin{align} 
 & \Cgain [W](k_1) = \frac{\pi}{2} \Ckint \delta(\underline{k}) \delta(\underline{\omega})
  \nonumber \\ & \quad \times
 \left( W_3 J[\tilde{W}_2 W_4]
+ J[ W_4 \tilde{W}_2] W_3 +
W_4 J[\tilde{W}_2 W_3]
+ J[ W_3 \tilde{W}_2] W_4 
\right) \, .
\end{align}
If $W\in \Xphys$, then $\tilde{W}_2,W_3,W_4\ge 0$ inside the integrand above.
Therefore, we can apply Lemma \ref{th:matrixpos} and conclude that the
integrand is pointwise a positive matrix.  This directly implies that $\Cgain
[W](k_1)\ge 0$, since then for any $\psi\in \C^2$ clearly 
$\spr{\psi}{\Cgain  [W](k_1)\psi}\ge 0$. 

If $W\in \Xherm$, then $\trunc{W}\in \Xphys$, and thus the second statement is
a corollary of the first one. 
\end{proof}

With the above preparations, we are now ready to prove the preservation of
positivity.  Fix thus some $W_0\in \Xherm$ with $W_0\ge 0$ and let $w$ denote
the corresponding global solution to the truncated problem. 
Set $h_t(k):=\Heff^\vep[\trunc{w_t}](k)$.  Then each $h_t(k)$ is a Hermitian
matrix, and since the map $t \mapsto h_t(k)$ is also norm continuous for each
$k$, the standard Dyson expansion techniques (see, e.g., \cite[Theorem
X.69]{reedsimonII} and take an adjoint of the result) imply that for any $k\in
\T^d$ and $s,t$, with $0\le s\le t$, we can find a unitary matrix $u_{t,s}(k)$
such that  $u_{t,t}(k)=1$, the map $s\mapsto u_{t,s}(k)$ belongs to
$C^{(1)}([0,t],\C^{2\times 2})$, and  
  \begin{align}
  & \partial_s u_{t,s}(k)=\ci u_{t,s}(k) h_s(k)\, .
  \end{align}

Given matrices $u_{t,s}(k)$ as above, let us define
\begin{align}\label{eq:defgts}
v_{t,s}(k) :=  u_{t,s}(k) w_s(k)  u_{t,s}(k)^*\, ,
\quad 0\le s\le t\, ,\ k\in \T^d\, .
\end{align}
By (\ref{eq:defCtr}) and (\ref{eq:defgts}) then for any $t>0$ and $0<s<t$
\begin{align}\label{eq:vevol}
& \partial_s v_{t,s}(k) =  u_{t,s}(k) \left(\ci h_s(k) w_s(k) 
 + \partial_s w_s(k) - \ci w_s(k) h_s(k) \right)   u_{t,s}(k)^*
\nonumber \\ & \quad 
= u_{t,s}(k) \Ccolltr [w_s](k) u_{t,s}(k)^*\nonumber \\ & \quad 
= g_{t,s}(k) - b_{t,s}(k)v_{t,s}(k) - v_{t,s}(k)b_{t,s}(k)^*\, .
\end{align}
where on the last step we have used the unitarity of $u_{t,s}(k)$ and
introduced the shorthand notations 
\begin{align}
 &  g_{t,s}(k) :=  u_{t,s}(k)\Cgain [\trunc{w_s}](k)u_{t,s}(k)^*\, ,
 \label{eq:ut}\\
 & b_{t,s}(k) :=  u_{t,s}(k)\Closs [\trunc{w_s}](k)u_{t,s}(k)^*\, . 
\end{align}
By Lemma \ref{th:posgain}, here $g_{t,s}(k)\ge 0$.  Fix for the moment $t>0$
and $k\in \T^d$.  Since the map $s\mapsto b_{t,s}(k)$ belongs to
$C([0,t],\C^{2\times 2})$, there is is unique solution $F\in
C^{(1)}([0,t],\C^{2\times 2})$, $0\le s\le t$, to the matrix equation
$\partial_s F(s)= -F(s)b_{t,t-s}(k)$, with initial data $F(0)=1$.  (The
solution can be obtained by a time-ordered exponential, similarly to
$u_{t,s}(k)$; more details about matrix equations of this type can be found
for instance in \cite{Bellman70}.)  Then $\partial_s F(s)^*= -b_{t,t-s}(k)^*
F(s)^*$ and (\ref{eq:vevol}) implies that 
\begin{align}%\label{eq:vevol}
& \partial_s (F(t-s) v_{t,s}(k) F(t-s)^*) = F(t-s) g_{t,s}(k) F(t-s)^*\,  .
\end{align}
Since $v_{t,t}(k)=w_t(k)$ and $v_{t,0}(k)=u_{t,0}(k) w_0(k)  u_{t,0}(k)^*$, we
find that 
\begin{align}%\label{eq:vevol}
& w_t(k) = (F(t)u_{t,0}(k)) W_0(k)  (F(t)u_{t,0}(k))^* + \int_0
^t\!\mathrm{d}s\, F(t-s) g_{t,s}(k) F(t-s)^*\, . 
\end{align}
As mentioned above, here  $g_{t,s}(k)\ge 0$ for all $s$ and, since by
assumption $W_0(k) \ge 0$, the above formula shows that $w_t(k)\ge 0$.  Since
$t$ and $k$ were arbitrary, we can conclude that $w\ge 0$ if $W_0\ge 0$. 

This shows that the truncated time-evolution preserves positivity which was
the missing part from the well-posedness result.  Therefore, we can now also
conclude that, if   $W_0\in \Xphys$, then  
$w_t\in \Xphys$ for all $t$, and $W_t:=w_t$ provides a solution to the
original evolution equation.  Integrating (\ref{eq:vevol}) over $s$ then
yields 
\begin{align}%\label{eq:vevol}
& W_t(k) = u_{t,0}(k) W_0(k)  u_{t,0}(k)^* + \int_0 ^t\!\mathrm{d}s\, 
u_{t,s}(k) \Ccolltr [W_s](k) u_{t,s}(k)^* \, .
\end{align}
Since $\Phi[W_s]=W_s$, here 
$\Ccolltr [W_s](k) = \Ccoll [W_s](k)$, and $u_{t,s}(k)$ satisfies 
$u_{t,t}(k)=1$ and $\partial_s u_{t,s}(k)=\ci u_{t,s}(k) \Heff^\vep[W_s](k)$.
Therefore, the second paragraph of the Theorem holds with the choice
$U^\vep_{t,s}(k;W_0):=u_{t,s}(k)$. 

It only remains to prove that the above solution also preserves energy and
spin.  For this, consider some $W_0\in \Xphys$, and let $w$ denote the
corresponding solution satisfying $\partial_t w_t(k_1) = \Ctot^\vep
[w_t](k_1)$ for all $t>0$, $k_1\in \T^d$.   
Then $w_t(k)$ satisfies (\ref{eq:fpeqn}) for all $t,k$ and this, together with
the uniform bounds in (\ref{eq:bd}), allows using Fubini's theorem in the
definitions of the conserved quantities.  Therefore, it is sufficient to check
that for all $W\in \Xphys$ 
\begin{align}\label{eq:consgoals}
  & \int_{\T^d}\!\rmd k_1\, \omega(k_1) \trace \Ctot^\vep [W](k_1)=0\, , \qquad
  \int_{\T^d}\!\rmd k_1\, \Ctot^\vep [W](k_1)=0 \, .
\end{align}

Let us begin with the first, scalar valued case, implying conservation of
energy.  First, by cyclicity of trace, $\trace \Ctot^\vep [W](k_1)=\trace
\Ccoll [W](k_1)$.  Since $W$ is continuous, we can evaluate the integral over
$\nu_{k_1}$ by using the formula (\ref{eq:numaindef}) where, to avoid
confusion, let us denote the new regularizing variable by $\vep_0$ instead of
$\vep$.  As shown in the proof of Proposition \ref{th:defC1}, the resulting
integral is uniformly bounded in $k_1$ and $\vep_0$, hence dominated
convergence can be applied to prove that $\int_{\T^d}\!\rmd k_1\, \omega(k_1)
\trace \Ctot^\vep [W](k_1)$ is equal to the $\vep_0\to 0^+$ limit of the
integral 
\begin{align}\label{eq:encons1}
& \int_{(\mathbb{T}^{d})^4}\!\rmd^4 k\, 
 \delta(\underline{k}) \frac{\vep_0}{\vep_0^2+\underline{\omega}^2} \omega_1
\nonumber \\  & \quad 
 \times  \tr\left( \tilde{W}_1 W_3 J[\tilde{W}_2 W_4]
+ J[ W_4 \tilde{W}_2] W_3 \tilde{W}_1 
 -W_1\tilde{W}_3 J[W_2\tilde{W}_4]
- J[\tilde{W}_4 W_2]\tilde{W}_3  W_1 \right) \, .
\end{align}
The trace-factor on the second line changes sign if we relabel the integration
variables by $k_1\leftrightarrow k_3$ and $k_2\leftrightarrow k_4$, and it is
invariant under the relabelling $k_1\leftrightarrow k_2$, $k_3\leftrightarrow
k_4$; these properties can be proven by using the cyclicity of the trace and
the definition of $J$ in (\ref{eq:defJ}).  As both relabellings leave
$|\underline{\omega}|$ and $|\underline{k}|$ invariant, by first taking the
average over the first swap and then the average of performing the second swap
to the result, we find that the value of (\ref{eq:encons1}) does not change if
we replace the factor $\omega_1$ by
$(\omega_1-\omega_3+\omega_2-\omega_4)/4=\underline{\omega}/4$ there.
However, then we can retrace the steps above, and produce an integral over
$\nu_{k_1}$ which contains a factor $\underline{\omega}$.  By item
\ref{it:fomint} in Proposition \ref{th:defC1}, the value of such an integral
is zero.  Hence we have shown that the first equality in (\ref{eq:consgoals})  
holds.

Let us then consider the second, matrix equality, in (\ref{eq:consgoals}).  We
use the split $\Ctot^\vep =\Ccoll + \Cdisp^\vep$ and show independently that
both of the resulting two terms are zero.  It is easier to check the
conservative term using the following integral representation, analogous to
the dissipative term:  
\begin{align}%\label{eq:consgoals}
& -\Ckint \delta(\underline{k}) 
 \frac{\underline{\omega}}{\underline{\omega}^2+\vep^2}
\nonumber \\  & \qquad 
\times 
  \left( \tilde{W}_1 W_3 J[\tilde{W}_2 W_4]
- J[ W_4 \tilde{W}_2] W_3 \tilde{W}_1 
 -W_1\tilde{W}_3 J[W_2\tilde{W}_4]
+ J[\tilde{W}_4 W_2]\tilde{W}_3  W_1 \right) 
\nonumber \\  & \quad 
= \Ckint \delta(\underline{k}) 
 \frac{\underline{\omega}}{\underline{\omega}^2+\vep^2}
  \left( W_3 \tilde{W}_2 W_4 - W_4 \tilde{W}_2 W_3 \right)
\nonumber \\  & \qquad 
- \Ckint \delta(\underline{k}) 
  \frac{\underline{\omega}}{\underline{\omega}^2+\vep^2}
\nonumber \\  & \qquad \quad
\times 
  \left( (J[ W_4 \tilde{W}_2] W_3 +J[\tilde{W}_4 W_2]\tilde{W}_3)W_1 
  - W_1 ( W_3 J[\tilde{W}_2 W_4]+\tilde{W}_3 J[W_2\tilde{W}_4])\right) 
\nonumber \\  & \quad 
=  [\Heff^\vep[W](k_1),W(k_1)]
\nonumber \\  & \quad 
= \ci \Cdisp^\vep[W](k_1)
\, ,
\end{align}
where in the second equality we have used the definition (\ref{eq:defHeffeps})
and the property that any term, whose integrand is antisymmetric under the
swap $k_3\leftrightarrow k_4$, evaluates to zero. 
Then we integrate the equality over $k_1$, use Fubini's theorem, and take an
average over the result from the swap $k_1\leftrightarrow k_3$,
$k_2\leftrightarrow k_4$, yielding 
\begin{align}\label{eq:dispint}
  & -\ci\int_{\T^d}\!\rmd k_1\, \Cdisp^\vep [W](k_1) \nonumber \\  & \quad
 = \frac{1}{2}\int_{(\mathbb{T}^{d})^4}\!\rmd^4 k\, 
 \delta(\underline{k}) \frac{\underline{\omega}}{\underline{\omega}^2+\vep^2}
\nonumber \\  & \qquad 
 \times  \Bigl( \tilde{W}_1 W_3 J[\tilde{W}_2 W_4]
- J[ W_4 \tilde{W}_2] W_3 \tilde{W}_1 
 -W_1\tilde{W}_3 J[W_2\tilde{W}_4]
+ J[\tilde{W}_4 W_2]\tilde{W}_3  W_1 
\nonumber \\  & \qquad \quad
-\tilde{W}_3 W_1 J[\tilde{W}_4 W_2]
+ J[ W_2 \tilde{W}_4] W_1 \tilde{W}_3
+ W_3\tilde{W}_1 J[W_4\tilde{W}_2]
- J[\tilde{W}_2 W_4]\tilde{W}_1  W_3
\Bigr) 
\, .
\end{align}
If we expand the definitions of $J$ in the integrand, all terms containing a
trace cancel out.  The remaining integrand is antisymmetric under the swap  
$k_1\leftrightarrow k_2$, $k_3\leftrightarrow k_4$, hence evaluates to zero.  

We have thus proven that $\int_{\T^d}\!\rmd k_1\, \Cdisp^\vep [W](k_1) =0 $.
The proof of $\int_{\T^d}\!\rmd k_1\, \Ccoll^\vep [W](k_1)=0$ follows by first
expressing the integral as a limit of terms analogous to (\ref{eq:encons1})
and then performing the swaps as in (\ref{eq:dispint}); we skip the details of
the computation here. 
This concludes the proof of the conservation laws, and thus also of the Theorem.
\end{proof}

\section{$L^2$-continuity of the collision operator}
\label{sec:pv}

In this section, we consider the regularized effective Hamiltonian starting
from (\ref{eq:defHeff2}), 
\begin{align}
& \Heff ^\vep[W](k_1) = -\frac{1}{2} \Ckint  \delta(\underline{k})\, 
\frac{\underline{\omega}}{\underline{\omega}^2 + \vep^2} \,
\Bigl( 2 \tr(W_2 \tilde{W}_4 + \tilde{W}_4 W_2 ) 1
\nonumber \\  & \qquad
+ 2 \tr(W_4-W_2) W_3 
+(W_2-W_4) W_3 + W_3 (W_2-W_4) - W_2 \tilde{W}_4 - \tilde{W}_4 W_2\Bigr)  \, ,
\end{align}
where $\underline{\omega}:= \omega(k_1)+\omega(k_2)-\omega(k_3)-\omega(k_4)$. 
By suitably integrating out the convolution $\delta$-function (i.e., by a
suitable change of integration variables) 
we can write its components as a finite sum of integrals of the same form.
Explicitly, define for  $k_1,k'_1,k'_2\in  \T^d$, 
\begin{align}
& \sigma\casen{2} :=(1,-1,-1,1)\, , \qquad
{k}\casen{2}:=(k_1,k'_1-k'_2-k_1,k'_1,-k'_2)\, , \\
&  \sigma\casen{3} :=(1,1,-1,-1)\, , \qquad
{k}\casen{3}:=(k_1,k'_1,k_1+k'_1-k'_2,k'_2)\, , \\
&  \sigma\casen{4} :=(1,1,-1,-1)\, , \qquad
{k}\casen{4}:=(k_1,k'_1,k'_2,k_1+k'_1-k'_2)\, , 
\end{align}
and recall the definition of $\omt(k;\sigma)$ in (\ref{eq:defomsig}).
Then, after performing the change of variables as listed in the definition of
$\vc{k}\casen{i}$, we find that $\underline{\omega}\to \omt((\sigma\casen{i}_2
k_1,k'_1,k'_2); \sigma\casen{i})$ inside the integrand for all $i=2,3,4$,
using the reflection invariance of $\omega$.  This shows that any component of
$\Heff^\vep [W](k_1)$ is a linear combination, with coefficients $\pm
\frac{1}{2}$, of a finite number of terms of the form 
\begin{align}
 \mathcal{I}^1_\vep[f,g](k'_0;\sigma) 
 :=  \int_{(\T^d)^2} \rmd k'_1 \rmd k'_2 \,f(k'_1) g(k'_2) 
  \frac{\omt }{\omt ^2+\vep^2}
\end{align}
where $\omt  = \omt((k'_0,k'_1,k'_2); \sigma)$ and in each term
$f(k):=W(k)_{ij}$, for some indices $i,j$, and $g$ is one of the following
three choices: $g=1$, $g(k)=W(k)_{i'j'}$, or $g(k)=W(-k)_{i'j'}$, for some
indices $i',j'$.  In addition, for each term there is an $i\in \set{2,3,4}$
such that $\sigma = \sigma\casen{i}$ and $k'_0 = \sigma_2 k_1$. 
Therefore, for each of the terms, $\norm{\mathcal{I}^1_\vep[f,g]}_2$ is equal
to the $L^2$-norm of the term (taken over $k_1$), and both of the functions
$f,g$ belong to $L^2(\T^d)$. 

The following results show that the assumptions
(DR\ref{it:DR1})--(DR\ref{it:DRpv}) suffice to make also the full collision
operator ``nicely'' continuous in $\Lherm$-norm.  For this, we also consider
the following approximate collision integrals with possibly discontinuous
input, 
\begin{align}
  \mathcal{I}^0_\vep[f,g](k'_0;\sigma) 
 := \int_{(\T^d)^2} \rmd k'_1 \rmd k'_2 \,f(k'_1) g(k'_2) 
  \frac{\vep}{\omt ^2+\vep^2} \, .
\end{align}

\begin{proposition}\label{th:I1conv}
  Assume that $\omega$ satisfies (DR\ref{it:DR1})--(DR\ref{it:DRpv}), and
  $\sigma\in \set{-1,1}^4$ is given.   Then for any  
 $f,g \in  L^\infty(\T^d)$ there are 
  $\mathcal{I}^1_0,\mathcal{I}^0_0\in L^2(\T^d)$ such that 
  $\mathcal{I}^0_\vep[f,g] \to  \mathcal{I}^0_0$,  $\mathcal{I}^1_\vep[f,g]
  \to  \mathcal{I}^1_0$ in $L^2$-norm as $\vep\to 0^+$.   For both $j=0,1$ and
  all $\vep\ge 0$, $\mathcal{I}^j_\vep$ is independent of the choice of
  representatives of $f,g$, and 
$\norm{\mathcal{I}^j_\vep}_{L^2} \le C \norm{f}_{L^2} \norm{g}_{L^\infty}$ 
 with $C:=  \frac{1}{2}C_{\mathcal{G}}^{1/4}$.  There is also a continuous
 function $u:[0,\infty)\to \R_+$ with $u(0)=0$ such that
 $\norm{\mathcal{I}^j_{\vep'}-\mathcal{I}^j_\vep}_{L^2} \le
 u(\max(\vep',\vep)) \norm{f}_{L^2} \norm{g}_{L^\infty}$ for every
 $\vep',\vep>0$, $f,g \in  L^\infty(\T^d)$, $j=0,1$.

 In addition,  $\mathcal{I}^0_0\in L^\infty(\T^d)$ with
 $\norm{\mathcal{I}^0_0}_\infty \le \pi \norm{\totcoll(\cdot,0;\sigma)}_\infty
 \norm{f}_\infty \norm{g}_\infty$, and 
 using the shorthand $\omt(k'_0,k'_1,k'_2)  := \omt((k'_0,k'_1,k'_2); \sigma)$,
  \begin{align}\label{eq:I1rep}
    \underset{\vep\to 0}{L^2{-}\lim} \int_{(\T^d)^2} \rmd k'_1 \rmd k'_2 \, 
 \frac{\cf(|\omt(\cdot,k'_1,k'_2)|>\vep)}{\omt(\cdot,k'_1,k'_2)} 
 f(k'_1) g(k'_2)
     = \mathcal{I}^1_0(\cdot) \, .
  \end{align}
  
  Suppose that $f_\vep,g_\vep \in L^\infty(\T^d)$, $0<\vep<1$, are such that
  $m:=\sup_\vep\max(\norm{f_\vep}_\infty,\norm{g_\vep}_\infty) <\infty$, and
  $f_\vep\to f$, $g_\vep \to g$  in $L^2$-norm as $\vep\to 0^+$.  Then $f,g\in
  L^\infty$ and  
  $\norm{\mathcal{I}^j_\vep[f,g]-\mathcal{I}^j_{\vep}[f_\vep,g_\vep]}_{L^2}\le
  C_{\mathcal{G}}^{1/4}\frac{m}{2}
  (\norm{f-f_\vep}_{L^2}+\norm{g-g_\vep}_{L^2})$ for $\vep>0$ and $j=0,1$.  In
  particular,  
  $\mathcal{I}^0_\vep[f_\vep,g_\vep] \to  \mathcal{I}^0_0[f,g]$ and
  $\mathcal{I}^1_\vep[f_\vep,g_\vep] \to  \mathcal{I}^1_0[f,g]$ in $L^2$-norm
  as $\vep\to 0^+$.  
\end{proposition}
The key result connecting the limits to the assumption (DR\ref{it:DRpv}) is
given by the following Lemma. 
\begin{lemma}\label{th:l2linfbound}
  Assume that $\omega$ is continuous and satisfies (DR\ref{it:DRpv}), and
  $\sigma\in \set{-1,1}^4$, $f,g \in  L^\infty(\T^d)$ are given.  Suppose
  $\varphi,\FT{\varphi}:\R \to \C$ are such that $\FT{\varphi} \in L^1\cap
  L^\infty$ and $\varphi(x) = \int_{-\infty}^\infty \!\rmd s\, \FT{\varphi}(s)
  \rme^{\ci s x}$ for all $x\in \R$.  Define for  $k_1\in \T^d$ 
  \begin{align}
  \mathcal{I}[f,g](k_1) := \int_{(\T^d)^2} \rmd k_2 \rmd k_3 \,f(k_2) g(k_3) 
  \varphi(\omt(k;\sigma))\, .
\end{align}
  Then $\mathcal{I}[f,g]$ is independent of the choice of the representatives
  for $f,g$ and 
  $\norm{\mathcal{I}[f,g]}_{L^2}\le C_0\norm{f}_{L^2} \norm{g}_{L^\infty}$ with 
  $C_0^4 := \int_{\R^4} \rmd s \prod_{i=1}^4 |\FT{\varphi}((-1)^{i-1} s_i)| 
  \left|\mathcal{G}(s;\sigma)\right|\le \norm{\FT{\varphi}}_\infty^4 
  C_{\mathcal{G}} <\infty$.
\end{lemma}
\begin{proof}
Let us drop $\sigma$ from the notation, and denote $\omt (k) := \omt(k;\sigma)$.
It follows from the assumptions that $\varphi$ is continuous, and thus the
integral in the definition of $\mathcal{I}[f,g]$ is always convergent, and its
value remains invariant if $f$ and $g$ are changed in a set of Lebesgue
measure zero.   We use Fubini's theorem to integrate over $k_3$ first, which
shows that 
\begin{align}
   \norm{\mathcal{I}[f,g]}_{L^2}^2 \le \norm{g}_\infty^2 
    \int_{\T^d} \rmd k_1 \left(
    \int_{\T^d} \rmd k_3 \left| \int_{\T^d} \rmd k_2 \,f(k_2) 
      \varphi(\omt (k)) \right| \right)^2\, .
\end{align}
By the Cauchy-Schwarz inequality and the normalization $\int_{\T^d} \rmd k=1$,
the remaining integral is bounded by 
\begin{align}
  \int_{(\T^d)^2}\! \rmd k_1 \rmd k_3
   \left| \int_{\T^d} \rmd k_2 \,f(k_2) \varphi(\omt (k)) \right|^2\, .
\end{align}
By Fubini's theorem, this is equal to
\begin{align}
  \int_{(\T^d)^2}\! \rmd k_2 \rmd k'_2 \,f(k'_2)^* f(k_2)
    \int_{(\T^d)^2}\! \rmd k_1 \rmd k_3 \, \varphi(\omt ) \varphi(\omt ')^*\, ,
\end{align}
where $\omt ':= \omt(k_1,k'_2,k_3;\sigma)$.  Thus it is bounded by
$\norm{f}_{L^2}^2$ times the square root of 
\begin{align}\label{eq:keyconst}
  \int_{(\T^d)^2}\! \rmd k_2 \rmd k'_2 \left|
    \int_{(\T^d)^2}\! \rmd k_1 \rmd k_3 \, 
 \varphi(\omt ) \varphi(\omt ')^*\right|^2\, .
\end{align}
Therefore, we can conclude that the Lemma holds if we can prove that
(\ref{eq:keyconst}) can be bounded by $C_0^4$. 
 
By assumption, we have $\varphi(x) = \int_{-\infty}^\infty \!\rmd s\,
\FT{\varphi}(s) \rme^{\ci s x}$ 
where $\FT{\varphi}\in L^1$.   Therefore, by Fubini's theorem 
  \begin{align}
  \int_{\R^4}\! \rmd s \, \FT{\varphi}(s_1) \FT{\varphi}(-s_2)^* 
  \FT{\varphi}(s_3) \FT{\varphi}(-s_4)^* \mathcal{G}(s;\sigma)
  = \int_{(\T^d)^3\times (\T^d)^3} \!\rmd^3 k' \rmd^3 k\, 
  \varphi(\Omega_1) \varphi(\Omega_2)^*\varphi(\Omega_3)\varphi(\Omega_4)^*\, ,
\end{align}
where we have used the same notations as in (DR\ref{it:DRpv}). However, a
second application of Fubini's theorem shows that the right hand side here is
equal to (\ref{eq:keyconst}).  This implies that (\ref{eq:keyconst}) is
bounded by $C_0^4$ which concludes the proof of the Lemma. 
\end{proof}

\begin{proofof}{Proposition \ref{th:I1conv}}
Recall the definition of $\varphi^0_\vep$ and $\FT{\varphi}^0_\vep$ in
(\ref{eq:defvarphi0}), and set also  
$\varphi^1_\vep(x):= \frac{x}{x^2+\vep^2}$ and
$\FT{\varphi}^1_\vep(x) := \frac{1}{2 \ci} \sign(x) \rme^{-\vep |x|}$, 
$x\in \R$, $\vep>0$, where $\sign(x)$ denotes the sign of $x$ which we choose
to be $0$ if $x=0$.   By direct integration, we find that both of the pairs
$(\varphi^0_\vep,\FT{\varphi}^0_\vep)$ and
$(\varphi^1_\vep,\FT{\varphi}^1_\vep)$ satisfy the assumptions of Lemma
\ref{th:l2linfbound} for any $\vep>0$. 
This immediately implies that for any $j=0,1$, $\vep>0$,
$\mathcal{I}^j_\vep[f,g]$ is independent of the choice of representatives and 
$\norm{\mathcal{I}^j_\vep[f,g]}_{L^2}\le
\frac{1}{2}C_{\mathcal{G}}^{1/4}\norm{f}_{L^2} \norm{g}_\infty<\infty$.
Obviously, these properties then also hold for any $L^2$-limit points. 

In fact, then also the pair $(\varphi,\FT{\varphi})$ satisfies assumptions of
the Lemma, if $j=0,1$, $0<\vep<\vep_0$, and we define  
$\varphi(x) := \varphi^j_\vep(x)-\varphi^j_{\vep_0}(x)$ and $\FT{\varphi}(x)
:= \FT{\varphi}^j_\vep(x)-\FT{\varphi}^j_{\vep_0}(x)$ for $x\in \R$.  Since
then  
$|\FT{\varphi}(s)|= \frac{1}{2} \rme^{-\vep |s|}
(1-\rme^{-(\vep_0-\vep)|s|})\le \frac{1}{2} (1-\rme^{-\vep_0|s|})$, the Lemma
implies that  
\begin{align}\label{eq:cauchybound}
 & \norm{\mathcal{I}^j_\vep[f,g]-\mathcal{I}^j_{\vep_0}[f,g]}_{L^2} 
  \le u(\vep_0) \norm{f}_{L^2} \norm{g}_\infty \, , 
\end{align}
where $u:[0,\infty)\to \R_+$ is defined by 
\begin{align}\label{eq:defubnd}
  u(\vep_0) := \frac{1}{2}
 \left(\int_{\R^4} \rmd s \prod_{i=1}^4 \left|1-\rme^{-\vep_0|s_i|}\right| 
  \left|\mathcal{G}(s;\sigma)\right|\right)^{\!\!\frac{1}{4}}\, .
\end{align}
Clearly, $u(0)=0$, and dominated convergence theorem can be applied here to
prove that $u$ is continuous. 
This proves the last statement in the first part of the Proposition.

We can conclude that 
$\norm{\mathcal{I}^j_\vep[f,g]-\mathcal{I}^j_{\vep'}[f,g]}_{L^2}\to 0$ when
$\max(\vep,\vep')\to 0^+$.
Therefore, if $\vep_n>0$, $n\in \N$, and $\vep_n\to 0$, then
$\mathcal{I}^j_{\vep_n}[f,g]$ is a Cauchy sequence in $L^2$.  Thus the
sequence converges in $L^2$, and using (\ref{eq:cauchybound}) it is
straightforward to check that the limit is independent of the choice of the
sequence $(\vep_n)$.  Let the unique limit point be denoted by
$\mathcal{I}^j_0[f,g]$.  Then the first two statements of the Proposition
hold.  In addition,   
taking  $\vep\to 0$ in (\ref{eq:cauchybound}) also shows that 
for any $\vep_0>0$ and both $j=0,1$
\begin{align}
  \norm{\mathcal{I}^j_0[f,g]-\mathcal{I}^j_{\vep_0}[f,g]}_{L^2} \le 
   u(\vep_0)\norm{f}_{L^2} \norm{g}_\infty 
  \, .
\end{align}

By the $L^2$-convergence we can find a sequence $\vep_n>0$, $n\in N$, such
that $\vep_n\to 0$ and  
$\mathcal{I}^0_{\vep_n}[f,g](k_0)\to \mathcal{I}^0_{0}(k_0)$ for almost every
$k_0$.  But since then 
\begin{align}
  |\mathcal{I}^0_{\vep_n}[f,g](k_0)|\le \norm{f}_\infty \norm{g}_\infty 
\int_{(\T^d)^2} \rmd k'_1 \rmd k'_2 \, 
\frac{\vep_n}{\omt ^2+\vep_n ^2} \overset{n\to\infty}{\longrightarrow} \pi
\totcoll(k_0,\sigma) \norm{f}_\infty \norm{g}_\infty\, , 
\end{align}
this implies also $|\mathcal{I}^0_{0}(k_0)|\le \pi \norm{\totcoll}_{\infty}
\norm{f}_\infty \norm{g}_\infty$. 
As the complement of such $k_0$ has Lebesgue measure zero, we have proven the
claim made about $\mathcal{I}^0_{0}$ in the Proposition. 

To prove (\ref{eq:I1rep}), consider
\begin{align}
\mathcal{J}^1_\vep(k_0) := \int_{(\T^d)^2} \rmd k_1 \rmd k_2 \, 
\frac{\cf(|\omt(k_0,k_1,k_2)|>\vep)}{\omt(k_0,k_1,k_2)} f(k_1) g(k_2)\, .  
\end{align}
Obviously, $|\mathcal{J}^1_\vep(k_0)|\le \vep^{-1}
\norm{f}_\infty\norm{g}_\infty$, and thus $\mathcal{J}^1_\vep\in L^2$.  We
claim that $\norm{\mathcal{J}^1_\vep-\mathcal{I}^1_{\vep^2}}_{L^2} \to 0$ as
$\vep\to 0^+$, which implies that $\mathcal{J}^1_\vep \to \mathcal{I}^1_{0}$
in $L^2$.  For any $x\in \R$, 
\begin{align}
  \frac{x}{x^2+\vep^4} - \frac{\cf(|x|>\vep)}{x} = 
  -\cf(|x|>\vep)\frac{\vep^2}{x} \frac{\vep^2}{x^2+\vep^4} 
  +\cf(|x|\le \vep)\frac{x}{x^2+\vep^4}\, .
\end{align}
The first term is bounded by $\vep^3/(x^2+\vep^4)$, which shows that 
\begin{align}
  |\mathcal{J}^1_\vep(k_0)-\mathcal{I}^1_{\vep^2}(k_0)-\mathcal{J}^2_\vep(k_0)|\le
  \vep  \norm{f}_\infty\norm{g}_\infty \int_{(\T^d)^2} 
 \rmd k_1 \rmd k_2 \,\frac{\vep^2}{\omt^2+\vep^4}
  \le \vep  \norm{f}_\infty\norm{g}_\infty  C'\, ,
\end{align}
where $C'$ is the constant introduced in the proof of Proposition
\ref{th:defC1}, and 
\begin{align}
\mathcal{J}^2_\vep(k_0) := \int_{(\T^d)^2} \rmd k_1 \rmd k_2 \, 
 \cf(|\omt|\le \vep)
\frac{\omt}{\omt^2+\vep^4} f(k_1) g(k_2)\, .  
\end{align}
Therefore, it suffices to prove that 
$\norm{\mathcal{J}^2_\vep}_{L^2}\to 0$ as $\vep\to 0^+$.

For this, define for $\vep,\delta>0$ the function 
$\FT{\varphi}^2_{\vep,\delta}:\R\to\C$ by
$\FT{\varphi}^2_{\vep,\delta}(s) := \frac{1}{\pi \ci} \rme^{-\delta |s|}
h_\vep(s)$ where 
\begin{align}
  h_\vep(s) := \int_0^\vep\!\rmd \alpha \frac{\alpha}{\alpha^2+\vep^4}
  \sin(\alpha s)\, . 
\end{align}
Clearly, $h_\vep\in L^\infty(\rmd s)$ and thus
$\FT{\varphi}^2_{\vep,\delta}\in L^1\cap L^\infty$ and we can define
$\varphi^2_{\vep,\delta}(x) = \int_{-\infty}^\infty \!\rmd s\,
\FT{\varphi}^2_{\vep,\delta}(s) \rme^{\ci s x}$, $x\in \R$.  Then Lemma
\ref{th:l2linfbound} can be applied to  
\begin{align}
\mathcal{J}^2_{\vep,\delta}(k_0) := \int_{(\T^d)^2} \rmd k_1 \rmd k_2 \,
 \varphi^2_{\vep,\delta}(\omt) f(k_1) g(k_2)\, ,
\end{align}
which shows that $\norm{\mathcal{J}^2_{\vep,\delta}}_{L^2}\le
c_\vep\pi^{-1}\norm{f}_{L^2} \norm{g}_{L^\infty}$ with  
$c_\vep^4 := \int_{\R^4} \rmd s \prod_{i=1}^4
\left|h_\vep(|s_i|)\right|\left|\mathcal{G}(s;\sigma)\right|$.   
On the other hand, explicit integration yields
for any $x\in \R$
\begin{align}
&\varphi^2_{\vep,\delta}(x) = -\int_{-\infty}^\infty \frac{\rmd s}{2\pi}
\int_{-\vep}^\vep\!\rmd \alpha \frac{\alpha}{\alpha^2+\vep^4} \rme^{\ci s
  (\alpha+x)-\delta |s|} 
\nonumber \\ & \quad 
= -\int_{-\vep}^\vep\!\frac{\rmd \alpha}{2\pi} \frac{\alpha}{\alpha^2+\vep^4}
\frac{2\delta}{(\alpha+x)^2+\delta^2}
\nonumber \\ & \quad
= \int_{-\infty}^\infty\!\frac{\rmd y}{\pi} \frac{1}{1+y^2} 
\cf(|x+\delta y|\le \vep)\frac{x+\delta y}{(x+\delta y)^2+\vep^4} \, , 
\end{align}
where in the last equality we have changed variables to
$y=-(\alpha+x)/\delta$.  By dominated convergence, for any $k\in (\T^d)^3$,
such that $|\omt(k)|\ne \vep$, we have $\lim_{\delta\to 0^+}
\varphi^2_{\vep,\delta}(\omt) = \cf(|\omt|\le
\vep)\frac{\omt}{\omt^2+\vep^4}$.  By Corollary \ref{th:omtlevelsets}, the set
of $k$ with  
$\omt(k)= \pm \vep$ has Lebesgue measure zero, and hence we find
$\int_{(\T^d)^3}\!\rmd^3 k |\varphi^2_{\vep,\delta}(\omt)-\cf(|\omt|\le
\vep)\frac{\omt}{\omt^2+\vep^4}|^2\to 0$ as $\delta\to 0^+$.  This implies
that $\mathcal{J}^2_{\vep,\delta}\to \mathcal{J}^2_\vep$ in $L^2(\T^2)$, and
thus the previous bounds prove that also $\norm{\mathcal{J}^2_{\vep}}_{L^2}\le
c_\vep\pi^{-1}\norm{f}_{L^2} \norm{g}_{L^\infty}$. 

Therefore, to conclude the proof of (\ref{eq:I1rep}), we only need to show
that $c_\vep\to 0$ when $\vep\to 0^+$.  Since $|\sin y|\le |y|$, we have
$|h_\vep(s)|\le \vep s$ for any $s\ge 0$, and thus $\lim_{\vep\to
  0^+}\prod_{i=1}^4 \left|h_\vep(|s_i|)\right|= 0$ for fixed $s\in \R^4$.  Now
if we can show that $|h_\vep(s)|$ is uniformly bounded for $s\ge 0$ and
$\vep>0$, assumption (DR\ref{it:DRpv}) allows to use dominated convergence to
prove that $c_\vep\to 0$.  If $0\le s\le 2\vep^{-1}$, we have already found
that $|h_\vep(s)|\le 2$.  Assume thus $s> 2\vep^{-1}$ and set $\alpha_0 :=
\frac{\pi}{2 s}\in(0,\vep)$.  Then 
$\left|\int_0^{\alpha_0}\!\rmd \alpha \frac{\alpha}{\alpha^2+\vep^4}
  \sin(\alpha s)\right|\le s \alpha_0<2$ and thus it remains to prove that the
integral over $[\alpha_0,\vep]$ is bounded.  By partial integration and using
$\cos(\alpha_0 s)=0$ we have 
\begin{align}
 & \int_{\alpha_0}^\vep\!\rmd \alpha \frac{\alpha}{\alpha^2+\vep^4}
 \sin(\alpha s)
 %\nonumber \\ & \quad
= -\frac{\vep}{\vep^2+\vep^4} \frac{1}{s} \cos(\vep s) +  
\frac{1}{s} \int_{\alpha_0}^\vep\!\rmd \alpha
\frac{\vep^4-\alpha^2}{(\alpha^2+\vep^4)^2} \cos(\alpha s)  
\, .
\end{align}
Here the first term is bounded by $\frac{1}{2}$ and the second by $\frac{1}{s}
\int_{\alpha_0}^\vep\!\rmd \alpha\, \alpha^{-2}\le \frac{1}{s\alpha_0}<1$.  We
have shown that always $|h_\vep(s)|\le 4$ which implies $c_\vep\to 0$ and
concludes the proof of (\ref{eq:I1rep}). 

To prove the second part of the Proposition, assume that $f_\vep,g_\vep \in
L^2(\T^d)$, $0<\vep<1$, are such that there is $m<\infty$ for which
$\norm{f_\vep}_\infty, \norm{g_\vep}_\infty \le m$, and $f_\vep\to f$, $g_\vep
\to g$  in $L^2$-norm as $\vep\to 0^+$.  Since then there is a sequence
$\vep_n\to 0^+$ such that $f_{\vep_n}\to f$, $g_{\vep_n} \to g$ pointwise
almost everywhere, this implies that also $\norm{f}_\infty, \norm{g}_\infty
\le m$, and thus the previous results can be applied. 
Telescoping 
$f(k_2) g(k_3)-f_\vep(k_2) g_\vep(k_3) = (f(k_2) -f_\vep(k_2)) g(k_3) +
f_\vep(k_2) (g(k_3)- g_\vep(k_3))$ and swapping the integration variables in
the second term we can resort to the above bounds and obtain 
$\norm{\mathcal{I}^j_\vep[f,g]-\mathcal{I}^j_{\vep}[f_\vep,g_\vep]}_{L^2}\le
C_{\mathcal{G}}^{1/4}\frac{m}{2}
(\norm{f-f_\vep}_{L^2}+\norm{g-g_\vep}_{L^2})$ for $\vep>0$.  This implies
also $\mathcal{I}^j_{\vep}[f_\vep,g_\vep] \to \mathcal{I}^j_0[f,g]$, as
claimed in the Proposition. 
\end{proofof}

\begin{corollary}\label{th:l2pv}
 Assume that $\omega$ satisfies (DR\ref{it:DR1})--(DR\ref{it:DRpv}) and $W\in
 \Lphys$.  Then there is a unique $\Heff[W]\in \Lherm$ such that
 $\Heff^\vep[W]\to \Heff[W]$ in norm as $\vep\to 0^+$.   
 Moreover, there is a sequence $\vep_n\to 0^+$ such that (\ref{eq:L2pv}) holds
 for almost every $k_1$. 
 
 In addition, there is a constant $C$ such that for any $W,W'\in \Lphys$ and
 $\vep,\vep'>0$ 
 \begin{align}
   &\norm{\Heff^\vep[W]-\Heff^{\vep'}[W']}_2\le 
  C(\norm{W'-W}_2 + u(\max(\vep',\vep)))\, , \label{eq:Heffepsbound}\\
   &\norm{\Heff[W]-\Heff[W']}_2\le C\norm{W'-W}_2 \, , \label{eq:Heffbound}
 \end{align}
 where $u$ denotes a function satisfying the conclusion of Proposition
 \ref{th:I1conv}. 
\end{corollary}
\begin{proof}
  As explained in the beginning of this section, any component of
  $\Heff^\vep[W]$ can be expressed as a finite linear combination of suitably
  chosen $\mathcal{I}^1_{\vep}[f,g]$-terms.  More precisely, it is
  straightforward to check that there is a finite index set $S$ such that for
  any $i,j\in \set{1,2}$ and $n\in S$ we can find a constant $c\in \R$, with
  $|c|\le 1$, $\ell \in \set{2,3,4}$,  
  $a,a',b,b'\in \set{1,2}$, and $\ell'\in \set{0,1,2}$ such that for all $k$
  and $W$ 
 \begin{align}\label{eq:heffvep}
   \Heff^\vep[W](k)_{ij} = 
  \sum_{n\in S} c \mathcal{I}^1_{\vep}[f_{a'b'},g^{(\ell')}_{ab}]
    (\sigma\casen{\ell}_2 k;\sigma\casen{\ell})\, ,
 \end{align}
 where $f_{a'b'}(k):= W(k)_{a'b'}$, $g^{(0)}_{ab}(k):= 1$, $g^{(1)}_{ab}(k):=
 W(k)_{ab}$, and $g^{(2)}_{ab}(k):= W(-k)_{ab}$.  Note that  
 the constants inside the sum can depend on all of $n,i,j$, even though we
 have suppressed the dependence from the notation. 
 
If $k$ is such that $0\le W(k)\le 1$, then $|W(k)_{ij}|\le 1$ for all $i,j$.
Therefore, in (\ref{eq:heffvep}) always $f,g\in L^\infty(\T^d)$ with
$\norm{f}_\infty,\norm{g}_\infty\le 1$.  Thus we can apply Proposition
\ref{th:I1conv} and conclude that $\Heff^\vep[W]_{ij}\to \Heff[W]_{ij}$ in
$L^2(\T^d)$ as $\vep\to 0^+$, where 
\begin{align}\label{eq:heffdef}
   \Heff[W](k)_{ij} := \sum_{n\in S} c \mathcal{I}^1_{0}[f_{a'b'},g^{(\ell')}_{ab}]
    (\sigma\casen{\ell}_2 k;\sigma\casen{\ell})\, .
 \end{align}
Let $\Heff[W](k)$ denote the matrix collecting the above limit functions.   
Then we can find a sequence $\vep_n\to 0$ such that
$\Heff^{\vep_n}[W](k)_{ij}\to \Heff[W](k)_{ij}$ for all $i,j$ and almost every
$k$. 
As $\Heff^\vep[W]_{ij}(k)^* = \Heff^\vep[W]_{ji}(k)$ for every $k$, this
implies that $\Heff[W](k)$ is then Hermitian at almost every $k$.  In
addition, 
$\norm{\Heff^\vep[W]-\Heff[W]}_2^2= \sum_{i,j}
\norm{\Heff^\vep[W]_{ij}-\Heff[W]_{ij}}_{L^2}^2$, 
and therefore $\Heff[W]\in \Lherm$ and $\Heff^\vep[W]\to \Heff[W]$ as $\vep\to
0$.  Since the index set $S$ is finite, the representation of the limit as a
standard principal value integral, (\ref{eq:L2pv}), follows then from an
application of (\ref{eq:I1rep}) to (\ref{eq:heffdef}). 

Suppose then that $W,W'\in \Lphys$ and $\vep,\vep'>0$.  Then 
$\norm{\Heff^\vep[W]-\Heff^{\vep'}[W']}_2 \le  
\norm{\Heff^\vep[W]-\Heff^{\vep}[W']}_2+
\norm{\Heff^{\vep}[W']-\Heff^{\vep'}[W']}_2$. 
Let $u$ be a function as in Proposition \ref{th:I1conv}.  We can conclude that
$\norm{\Heff^{\vep}[W']_{ij}-\Heff^{\vep'}[W']_{ij}}_2\le |S|
u(\max(\vep,\vep')) \norm{W'}_2$ for any $i,j$. 
Here $\norm{W'}_2^2\le 2$, since $W'\in \Lphys$, and thus 
$\norm{\Heff^{\vep}[W']-\Heff^{\vep'}[W']}_2\le 4 |S| u(\max(\vep,\vep'))$.

To study $\norm{\Heff^\vep[W]-\Heff^{\vep}[W']}_2$ write $\Delta:= W-W'$ and
note that 
then $f_{ab}[W]=f_{ab}[W']+f_{ab}[\Delta]$,
$g^{(\ell')}_{ab}[W]=g^{(\ell')}_{ab}[W']+g^{(\ell')}_{ab}[\Delta]$, for
$\ell'=1,2$, and $g^{(0)}_{ab}[W]=g^{(0)}_{ab}[W']$.  Since 
$\mathcal{I}^1_{\vep}[f,g]$ is obviously linear in both $f$ and $g$, this
shows that  
\begin{align}
   & \Heff^\vep[W](k)_{ij}-\Heff^\vep[W'](k)_{ij}
   = \sum_{n\in S} c \mathcal{I}^1_{\vep}[\Delta_{a'b'},g^{(\ell')}_{ab}[W]]
    (\sigma\casen{\ell}_2 k;\sigma\casen{\ell})
    \nonumber \\ & \quad
     + \sum_{n\in S} c \cf(\ell'\ne 0)
 \mathcal{I}^1_{\vep}[W'_{a'b'},g^{(\ell')}_{ab}[\Delta]]
    (\sigma\casen{\ell}_2 k;\sigma\casen{\ell})
     \, .
 \end{align}
Consider one term in the second sum here.  If $\ell'=0$, then the term is
zero. 
If $\ell'=1$, we swap the integration variables, and if $\ell'=2$, we change
the integration variables to $k''_1=-k'_2$ and $k''_2=-k'_1$.  This shows that
each term in the second sum is either zero or of the form  
$c \mathcal{I}^1_{\vep}[\Delta_{ab},g^{(\ell')}_{a'b'}[W']](\pm k;\sigma')$ 
for some choice of the sign and $\sigma'$.  Applying Proposition
\ref{th:I1conv}, this shows that  
$\norm{\Heff^\vep[W](k)_{ij}-\Heff^\vep[W'](k)_{ij}}_{L_2} \le |S|
C_{\mathcal{G}}^{1/4} \norm{\Delta}_2$. 
Therefore, $\norm{\Heff^\vep[W]-\Heff^\vep[W']}_{2} \le 2 |S|^{1/2}
C_{\mathcal{G}}^{1/4} \norm{W-W'}_2$.

Thus by choosing $C$ as larger of the constants in these two bounds, we can
conclude that (\ref{eq:Heffepsbound}) holds. 
Then taking $\vep,\vep'\to 0$ in (\ref{eq:Heffepsbound}) proves also
(\ref{eq:Heffbound}) and concludes the proof of the corollary. 
\end{proof}

Finally, let us show that adding the assumption (DR\ref{it:DRpv}) simplifies
the definition of the dissipative term. 
\begin{corollary}\label{th:simpleccoll}
 Assume that $\omega$ satisfies (DR\ref{it:DR1})--(DR\ref{it:DRpv}).  Then for
 any $w_i \in L^\infty(\T^d)$, $i=1,2,3$, we have 
 \begin{align}\label{eq:c0lim2}
& \mathcal{C}_0[w_1,w_2,w_3](k_1) = \lim_{\vep\to 0^+} \int_{(\mathbb{T}^{d})^3}\!
 \frac{\rmd k_2\rmd k_3 \rmd k_4}{\pi} \delta(\underline{k})\,
  \frac{\vep}{\vep^2+\underline{\omega}^2} w_1(k_2) w_2(k_3) w_3(k_4)\, ,
\end{align}
where the limit converges in $L^2(\T^d)$-norm.

Therefore, then for any $W\in \Lphys$ and $i,j\in\set{1,2}$ the $L^2$-limit in
(\ref{eq:collop4}) holds. 
\end{corollary}
\begin{proof}
  For the given $w_i$, set $m:= \max \norm{w_i}_\infty$, and choose a sequence
  $v_{i,n}$ as explained at the end of the proof of Lemma
  \ref{th:l2colloper0}.  In particular, then $v_{i,n}\to w_i$ and
  $\mathcal{C}_{0,n}\to \mathcal{C}_0$ 
   in $L^2$-norm as $n\to\infty$, where
   $\mathcal{C}_{0,n}:=\mathcal{C}_0[v_{1,n},v_{2,n},v_{3,n}]$ and
   $\mathcal{C}_0:= \mathcal{C}_0[w_1,w_2,w_3]$. 
  For any $\vep>0$ and $k_1$ denote the value of the integral on the right
  hand side of (\ref{eq:c0lim2}) by $\mathcal{C}_\vep(k_1)$, and define  
  $\mathcal{C}_{\vep,n}:= \mathcal{C}_\vep[v_{1,n},v_{2,n},v_{3,n}]$ analogously.
    Then the function $\mathcal{C}_\vep\in L^\infty$, and we need to prove
    that $\norm{\mathcal{C}_0-\mathcal{C}_\vep}_2\to 0$ as $\vep\to 0^+$. 

For any $n,\vep$ we can estimate $\norm{\mathcal{C}_0-\mathcal{C}_\vep}_2 \le
\norm{\mathcal{C}_0-\mathcal{C}_{0,n}}_2 +
\norm{\mathcal{C}_{0,n}-\mathcal{C}_{\vep,n}}_2 +
\norm{\mathcal{C}_{\vep,n}-\mathcal{C}_\vep}_2$.  By Lemma
\ref{th:l2colloper0}, here $\norm{\mathcal{C}_0-\mathcal{C}_{0,n}}_2 \le C m^2
\sum_{i=1}^3 \norm{w_i-v_{i,n}}_2$.  Using a similar telescoping estimate and
Proposition \ref{th:I1conv}, we can also find a constant $C'$ such that
$\norm{\mathcal{C}_{\vep,n}-\mathcal{C}_\vep}_2\le C' m^2 \sum_{i=1}^3
\norm{w_i-v_{i,n}}_2$. 
Hence, for any $\vep_0>0$ we can find $n$ such that
$\norm{\mathcal{C}_0-\mathcal{C}_\vep}_2 \le 
\frac{\vep_0}{3} + \norm{\mathcal{C}_{0,n}-\mathcal{C}_{\vep,n}}_2$.  However,  
by Proposition \ref{th:defC1}, we can use dominated converge to conclude that
$\norm{\mathcal{C}_{0,n}-\mathcal{C}_{\vep,n}}_2\to 0$ as $\vep\to 0^+$.  In
particular, there is $\vep'>0$ such that for all $0<\vep<\vep'$ it is less
than $\frac{\vep_0}{3}$.  This proves that
$\norm{\mathcal{C}_0-\mathcal{C}_\vep}_2\to 0$ as $\vep\to 0^+$. 

Suppose then that $W\in \Lphys$ and $i,j\in\set{1,2}$ and consider the
definition of $\Ccoll[W]_{ij}$ given in (\ref{eq:collop2b}).  Each of the
$\mathcal{C}_0$-terms in the finite sum can be approximated in $L^2$-norm by
the corresponding $\mathcal{C}_\vep$-integral, for some fixed $\vep>0$.  The
resulting  $\mathcal{C}_\vep$-terms then sum to the right hand side of
(\ref{eq:collop4}), and hence the limit holds in $L^2$-norm. 
\end{proof}

\begin{proofof}{Theorems \ref{th:maindefC} and \ref{th:maindefHeff}}
Proposition \ref{th:defC1} and Corollary \ref{th:l2colloper} imply the
statements in the first paragraph of Theorem \ref{th:maindefC}, as well as the
first statement of the second paragraph.  Corollary \ref{th:simpleccoll}
proves the remaining statements in the second paragraph and completes the
proof of Theorem \ref{th:maindefC}. 

Theorem \ref{th:maindefHeff} is a direct consequence of Corollary
\ref{th:l2pv}. 
\end{proofof}

\section{Proof of Theorem \ref{th:main1}}
\label{sec:provemain}

Assume now that all of the conditions (DR\ref{it:DR1})--(DR\ref{it:DRpv})
hold, and define $\Ccoll$ and $\Heff$ as in the already proven Theorems
\ref{th:maindefC} and \ref{th:maindefHeff}. 
For any representative of $W\in \Lphys$ and $0<\delta<1$ let us define the
corresponding ``locally averaged representative'' $A_\delta[W]$ by setting for
$k\in \T^d$, 
\begin{align}
  A_\delta[W](k) := \frac{1}{Z_{\delta}} \int_{\T^d} \!\rmd k' \, 
 \cf(|k'|\le \delta) W(k+k')\, ,
\end{align}
where $Z_\delta:= \int_{\T^d} \!\rmd k \, \cf(|k|\le \delta)$ is the
appropriate normalization factor.  The formula is understood as a
vector-valued integral over the compact set $\defset{k\in\T^d}{|k|\le \delta}$
in the Banach space $\Lherm$.  Since $\Lphys\subset L^\infty(\T^d,\C^{2\times
  2})$, we can use dominated converge to conclude that $A_\delta[W]$ is
continuous on $\T^d$.  Also the Fermi property is clearly preserved and thus
$A_\delta$ is a linear map from $\Lphys$ to $\Xphys$.   

We can also identify any component function $W_{ij}(k)$ with a function in
$L^1(\R^d)$ by extending it periodically to the neighboring ``cells'' and
then setting the extension to be zero in all other cells. If $k$ is a Lebesgue
point of this extension, then $A_\delta[W]_{ij}(k)\to W_{ij}(k)$ as $\delta\to
0$, cf.\  \cite[Theorem 7.10]{rudin:rca}.  Since almost every point is then a
Lebesgue point, we find that $\lim_{\delta\to 0}\norm{A_\delta[W](k)- W(k)}^2=
0$ apart possibly from a set of measure zero collecting the non-Lebesgue
points of the three independent component functions of $W$.  Since
$\norm{A_\delta[W]}_\infty \le \norm{W}_\infty <\infty$, dominated convergence
implies that $A_\delta[W]\to W$ in $\Lherm$. 

Let us then consider some allowed initial data $W_0\in \Lphys$ and some
allowed regulators $0<\vep,\delta<1$.  By Theorem \ref{th:main2}, for all
$t\ge 0$, we then have a unique solution  
$w_{t}(\cdot;\vep,\delta,W_0)\in \Xphys$ to the regularized problem with
initial data $w_{0}=A_\delta[W_0]$.  We can also find for  $k\in \T^d$ and
$0\le s\le t$ unitary matrices $u_{t,s}(k;\vep,\delta,W_0)$ such that
$u_{t,t}(k)=1$ and $s\mapsto u_{t,s}(k)$ belongs to $C^{(1)}([0,t],\C^{2\times
  2})$ with $\partial_s u_{t,s}(k)=\ci u_{t,s}(k) \Heff^\vep[w_s](k)$, and we
have for all $t\ge 0$, and $k$, 
\begin{align}\label{eq:WtwithU2}
  & w_t(k) = u_{t,0}(k) w_0(k) u_{t,0}(k)^* 
  % \nonumber \\ & \quad 
   + \int_0^t\!\rmd s\, u_{t,s}(k) \Ccoll[w_s](k) u_{t,s}(k)^*\, .
\end{align}
 
Let us then consider two allowed initial data $W'_0,W_0\in \Lphys$ and some
allowed regulators $0<\vep',\vep,\delta',\delta<1$.  Define $w$ and $u$ as
above using $\vep$, $\delta$ and $W_0$, and set also
$w'_t(k):=w_t(k;\vep',\delta',W'_0)$ and  
$u'_{t,s}(k):=u_{t,s}(k;\vep',\delta',W'_0)$.  Using (\ref{eq:WtwithU2}) for
both solutions and telescoping shows that then  
\begin{align}\label{eq:diffwt}
  & w_t(k)-w'_t(k) = u_{t,0}(k) \left( \delta w_0(k) +  
  \delta u_{t,0}(k) w'_0(k) \right) u_{t,0}(k)^* 
   - u'_{t,0}(k) w'_0(k) \delta u_{t,0}(k)  u'_{t,0}(k)^* 
   \nonumber \\ & \quad 
   + \int_0^t\!\rmd s\, u_{t,s}(k)\left( \Ccoll[w_s](k)-
  \Ccoll[w'_s](k)\right)  u_{t,s}(k)^*
    \nonumber \\ & \quad 
   + \int_0^t\!\rmd s\, u_{t,s}(k)  \delta u_{t,s}(k)  
  \Ccoll[w'_s](k) u_{t,s}(k)^* 
   - \int_0^t\!\rmd s\, u'_{t,s}(k) \Ccoll[w'_s](k) 
  \delta u_{t,s}(k)  u'_{t,s}(k)^* 
\, ,
\end{align}
where $\delta u_{t,s}(k):= 1-u_{t,s}(k)^* u'_{t,s}(k)$ and $\delta w_0(k) :=
A_\delta[W_0]- A_{\delta'}[W'_0]$.  Clearly, $\norm{U M U^*}=\norm{M}$ if $U$
is a unitary matrix.  Thus we obtain from  
(\ref{eq:diffwt}) a bound
\begin{align}\label{eq:diffwt2}
  & \norm{w_t(k)-w'_t(k)} \le \norm{\delta w_0(k)} + 2\norm{\delta u_{t,0}(k)}
  \norm{w'_0(k)} 
   \nonumber \\ & \quad 
   + \int_0^t\!\rmd s\, \norm{\Ccoll[w_s](k)-\Ccoll[w'_s](k)}
   + 2 \int_0^t\!\rmd s\, \norm{\delta u_{t,s}(k)} \norm{\Ccoll[w'_s](k)}
\, .
\end{align}
Here $\norm{w'_0(k)}= \norm{A_{\delta'}[W'_0](k)}\le
\esssup_{k'}\norm{W'_0(k')}\le 2$, since $0\le W'_0\le 1$ almost everywhere.
Therefore, we find a bound 
\begin{align}
  & \norm{w_t-w'_t}_2 \le \norm{\delta w_0}_2 + 4\norm{\delta u_{t,0}}_2
   + \int_0^t\!\rmd s\, \norm{\Ccoll[w_s]-\Ccoll[w'_s]}_2
  \nonumber \\ & \quad 
   + 2 \int_0^t\!\rmd s\, \norm{\delta
     u_{t,s}}_2 \esssup_{k}\norm{\Ccoll[w'_s](k)} 
\, .
\end{align}
Here we can apply Corollary \ref{th:l2colloper}, which implies that there is a
pure constant $C$ such that  
\begin{align}\label{eq:diffwt3}
  & \norm{w_t-w'_t}_2 \le \norm{\delta w_0}_2 + 4\norm{\delta u_{t,0}}_2
   + C \int_0^t\!\rmd s\, \norm{w_s-w'_s}_2
  %\nonumber \\ & \quad 
   + 2 C \int_0^t\!\rmd s\, \norm{\delta u_{t,s}}_2
\, .
\end{align}

For simplicity, 
denote $h_s(k):=\Heff^\vep[w_s](k)$ and $h'_s(k):=\Heff^{\vep'}[w'_s](k)$.
Now $\delta u_{t,t}(k)=0$ and $\norm{\delta u_{t,s}(k)}^2 = 2 \tr 1 - 
\tr (u_{t,s}(k) u'_{t,s}(k)^* +u'_{t,s}(k) u_{t,s}(k)^* )$.  Thus 
$\partial_s \norm{\delta u_{t,s}(k)}^2 = -\ci
\tr\left[(h_s(k)-h'_s(k))u'_{t,s}(k)^*u_{t,s}(k)\right] 
+\ci \tr\left[(h_s(k)-h'_s(k))u_{t,s}(k)^* u'_{t,s}(k)\right]$.  Applying
Cauchy-Schwarz inequality to the Hilbert-Schmidt scalar product here, we find
that $|\partial_s \norm{\delta u_{t,s}(k)}^2|\le  
\norm{u'_{t,s}(k)^*u_{t,s}(k)-u_{t,s}(k)^* u'_{t,s}(k)}
\norm{h_s(k)-h'_s(k)}\le 2 \norm{\delta u_{t,s}(k)} \norm{h_s(k)-h'_s(k)}$.
Therefore, for any $0\le r\le t$ we have 
$\norm{\delta u_{t,r}(k)}^2 = - \int_r^t\! \rmd s\, \partial_s 
\norm{\delta u_{t,s}(k)}^2 \le 
2 \int_r^t\! \rmd s\,\norm{\delta u_{t,s}(k)} \norm{h_s(k)-h'_s(k)}$.  Set
$m_t(k):=\sup_{0\le s\le t} \norm{\delta u_{t,s}(k)}$, where obviously
$m_t(k)\le 4<\infty$. It follows that $\norm{\delta u_{t,r}(k)}^2 \le 2 m_t(k)
\int_0^t\! \rmd s\,\norm{h_s(k)-h'_s(k)}$, and hence $m_t(k)^2\le 2 m_t(k)
\int_0^t\! \rmd s\,\norm{h_s(k)-h'_s(k)}$.   Thus we can conclude that for all
$0\le r\le t$
\begin{align}
  \norm{\delta u_{t,r}(k)} \le 2  \int_0^t\! \rmd s\,
  \norm{\Heff^\vep[w_s](k)-\Heff^{\vep'}[w'_s](k)}\, .
\end{align}
Therefore, by using the constant $C$ and the function $u(x)$ as in Corollary
\ref{th:l2pv}, it follows that
$\norm{\delta u_{t,r}}_2 \le 2 C  \int_0^t\! 
\rmd s\,(\norm{w_s-w'_s}_2+u(\max(\vep,\vep')))$.  Applied in
(\ref{eq:diffwt3}), 
this shows that there is a pure constant $C'$ such that 
\begin{align}\label{eq:diffwt4}
  & \norm{w_t-w'_t}_2 \le \norm{\delta w_0}_2 + C' t (1+t) u(\max(\vep,\vep'))
  + C' (1+t) \int_0^t\!\rmd s\, \norm{w_s-w'_s}_2 
\, .
\end{align}
Since the map $t\mapsto \norm{w_t-w'_t}_2$ is continuous, Gr\"onwall's lemma
can be applied here, and we can conclude that, if $t_0>0$, then for all $0\le
t\le t_0$ we have 
\begin{align}\label{eq:WtGron}
  & \norm{w_t-w'_t}_2 \le \left( \norm{A_\delta[W_0]-
A_{\delta'}[W'_0]}_2 + C' t (1+t) u(\max(\vep,\vep')) \right) 
  \rme^{C' (1+t_0) t}
\, .
\end{align}

We can thus apply this also in the special case $W'_0=W_0$ and consider the
sequence of solutions defined for $\vep=1/n$, $\delta=1/n$: set $W_{t,n}(k):=
w_{t}(k;1/n,1/n,W_0)$, $n\in \N$.  By (\ref{eq:WtGron}), if $0\le t\le t_0$,
then  
\begin{align}
\norm{W_{t,n}-W_{t,n'}}_2\le 
\left(\norm{A_{1/n}[W_0]-A_{1/n'}[W_0]}_2 + C' t_0 (1+t_0)
  u(1/\min(n,n'))\right)\rme^{C' (1+t_0) t_0}\, .   
\end{align}
Thus $W_{t,n}$ forms a Cauchy sequence in $\Lherm$ and there exists $W_t :=
\lim_{n\to\infty} W_{t,n}$ for all $t\ge 0$ (by the previous discussion, the
limit indeed coincides with $W_0$ if $t=0$).  Since $W_{t,n}\in \Lphys$ for
all $n$ and $\Lphys$ is a closed subset, this implies that $W_t\in \Lphys$.
By Corollary \ref{th:l2pv}, then also  
\begin{align}
  \Heff[W_t] = \lim_{n\to \infty} \Heff^{\frac{1}{n}}[W_{t}]= 
  \lim_{n\to \infty} \Heff^{\frac{1}{n}}[W_{t,n}]\, ,
\end{align}
where the limits are taken in $L^2$-norm.  The conservation laws
(\ref{eq:claw1}) and (\ref{eq:claw2}) can be identified as scalar products in
$\Lherm$, the first with the function $k\mapsto \omega(k)1$, and the second
with the constant matrix $M_{i'j'}:=\delta_{ij'}\delta_{ji'}$.  Since the
equalities hold for all $W_{t,n}$ and these converge in $\Lherm$ to $W_t$,
this implies (\ref{eq:claw1}) and (\ref{eq:claw2}).

Since $W_{t,n}$ is a solution to the regularized problem, it satisfies a
pointwise identity 
\begin{align}\label{eq:wtdeDuh}
  & W_{t,n}(k) = W_{0,n}(k) + \int_0^t\! \rmd s\, \left(\Ccoll [W_{s,n}](k) 
  - \ci [\Heff^{1/n} [W_{s,n}](k),W_{s,n}(k)]\right)\, .
\end{align}
By taking a scalar product with an arbitrary element in $\Lherm$ and using
Corollaries  \ref{th:l2colloper}  and \ref{th:l2pv}, we thus find that as
vector valued integrals in $\Lherm$ 
\begin{align}\label{eq:wtdeDuh2}
  & W_{t} = W_{0} + \int_0^t\! \rmd s\, \left(\Ccoll[W_s] 
  - \ci [\Heff[W_s],W_s]\right)\, .
\end{align}
This implies first that $t\mapsto W_t$ is in $C([0,\infty),\Lphys)$ and the
Fr\'echet derivative satisfies $\partial_t W_t = \Ctot [W_t]$ for $t>0$.  Thus
$W\in C^{(1)}([0,\infty),\Lphys)$.

Finally, suppose $W'_t$ is an arbitrary solution in
$C^{(1)}([0,\infty),\Lphys)$ with initial data $W'_0$, and denote
$\Delta_t:=W_t-W'_t$.  Then $\Delta_t\in C^{(1)}([0,\infty),\Lherm)$ and 
\begin{align}
\partial_t \Delta_t = \Ccoll[W_t]- \Ccoll[W'_t] 
  - \ci [\Heff[W_t]-\Heff[W'_t],W'_t] -\ci [\Heff[W_t],\Delta_t]\, .
\end{align}
For this computation, and those that follow, it can be helpful to note that
$\Lherm$ is a real Hilbert space whose scalar product  
can be written as $(w',w)=\int\!\rmd k \tr(w'(k)w(k))$.  In particular, it can
be applied to conclude that $t\mapsto \norm{\Delta_{t}}^2_2$ is continuously
differentiable with 
$\partial_t \norm{\Delta_{t}}_2^2= 2 (\Delta_t,\partial_t \Delta_t)$. 
Since by periodicity of trace we have $\tr(\Delta_t(k)
[\Heff[W_t](k),\Delta_t(k)])=0$ for all $k$ , this shows that  
\begin{align}
\partial_t \norm{\Delta_{t}}_2^2= 2 \left(\Delta_t,\Ccoll[W_t]-
   \Ccoll[W'_t] - \ci [\Heff[W_t]-\Heff[W'_t],W'_t]\right)\, .
\end{align}
Applying Corollaries \ref{th:l2colloper}  and \ref{th:l2pv}, the property
$\norm{W'_t}\le 4$, and Cauchy-Schwarz inequality, we can conclude that there
is a constant $C$ such that  $\left|\partial_t \norm{\Delta_{t}}_2^2\right|
\le  C \norm{\Delta_t}_2^2$. 
Therefore, Gr\"onwall's lemma implies that for all $t\ge 0$
\begin{align}
  \norm{W_t-W'_t}_2^2 \le \norm{W_0-W'_0}_2^2 \rme^{C t}\, .
\end{align}
This result implies both the stated stability and uniqueness of solutions in
$\Lphys$, and thus concludes the proof of the Theorem. \qed

\appendix

\section{Nearest neighbor dispersion relation 
 at $d\ge 3$}\label{sec:hubbd3}
\label{sec:dispnn}

We prove here that the nearest neighbor dispersion relation of the square
lattice  
of dimension $d\ge 3$ satisfies all of the assumptions of the main theorem.
\begin{proposition}
If $d\ge 3$, all the properties listed in Assumption \ref{th:disprelass} are 
satisfied by the nearest neighbor dispersion relation, 
  \begin{align}\label{eq:defomnn}
  \omega(k):= c - \sum_{\nu=1}^d \cos p^\nu, \quad \text{with }p=2\pi k,
\end{align}
where $c\in \R$ is arbitrary.
\end{proposition}
\begin{proof}
Fix some $d\ge 3$ and define $\omega$ by (\ref{eq:defomnn}). 
It is clear that $\omega$ is then continuous and satisfies
$\omega(-k)=\omega(k)$.  In the Appendix to \cite{NLS09} it is already proven
that then there is a constant $C$ such that the free propagator $p_t(x)$
satisfies  
$\norm{p_t}^3 \le C (1+|t|)^{- \frac{3 d}{7}}$, and hence it belongs to
$L^1(\rmd t)$.  Therefore, (DR\ref{it:DR1}) and (DR\ref{it:DRdisp}) hold and  
we only need to prove that  also (DR\ref{it:DRpv}) is satisfied (note that the
present conditions are different from those defined in \cite{NLS09}). 

Fix $\sigma\in \set{-1,1}^4$ and recall the definition of the functions $\omt$
and $\Omega_i$ given in the statement of (DR\ref{it:DRpv}).   We need to
inspect 
\begin{align}
  \mathcal{G}(s) = \int_{(\T^d)^3\times (\T^d)^3} \!\rmd^3 k' \rmd^3 k\, 
  \rme^{\ci \sum_{i=1}^4 s_i \Omega_i(k,k')}\, ,
\end{align}
defined for $s\in \R^4$.  It suffices to prove that $\int_{\R^4} \rmd s
|\mathcal{G}(s)| < \infty$, as we can then choose as $C_{\mathcal{G}}$ the
maximum of all such bounds obtained from the 16 possible choices of $\sigma$. 

For the nearest neighbor dispersion relation, the constant term produces only
a global phase factor, and the integrals corresponding to the $d$ different
``directions'' factorize.  This shows that  
$\left|\mathcal{G}(s)\right| = |F(s)|^d\le |F(s)|^3$, where 
\begin{align}
& F(s) :=  
  \int_{\T^3\times \T^3} \!\rmd^3 k' \rmd^3 k\, \rme^{-\ci g(s,k,k')}\, , \\
& g(s,k,k') :=   \sum_{i,j=1}^4 s_i \sigma_j \cos(P_{ij}(k,k'))
 = \re\left(\sum_{i,j=1}^4 s_i \sigma_j \rme^{\ci P_{ij}(k,k')}\right)\, .
\end{align}
Each $P_{ij}$ is a linear function of $(k,k')$, which are easiest to define by
using the following matrix representation 
\begin{align}
 P(k,k') := 2\pi \begin{pmatrix}
    k_1 & k_2 & k_3 & k_1+k_2-k_3 \\
    k_1 & k'_2 & k_3 & k_1+k'_2-k_3 \\
    k'_1 & k_2 & k'_3 & k'_1+k_2-k'_3 \\
    k'_1 & k'_2 & k'_3 & k'_1+k'_2-k'_3 
  \end{pmatrix} \, .
\end{align}

In order to better decouple the interdependence, let us change the integration
variables from $(k,k')$ to $(q,\alpha)$ where $q_1=2\pi k_1$, 
 $q_2=2\pi k'_1$, $q_3=2\pi k_2$, 
$\alpha_1= 2 \pi (k_3-k_2)$, $\alpha_2=2\pi(k'_3-k'_2)$ and 
  $\alpha_3= 2 \pi (k_2-k'_2)$.
Then 
\begin{align}
 P(k,k') = \begin{pmatrix}
    q_1 & q_3 & q_3 + \alpha_1 & q_1-\alpha_1 \\
    q_1 & q_3-\alpha_3 & q_3+\alpha_1 & q_1-\alpha_1-\alpha_3 \\
    q_2 & q_3 & q_3 + \alpha_2-\alpha_3 & q_2+\alpha_3-\alpha_2 \\
    q_2 & q_3-\alpha_3 & q_3 + \alpha_2-\alpha_3 & q_2-\alpha_2 
  \end{pmatrix} \, ,
\end{align}
and, therefore,
\begin{align}
 & g(s,k,k') 
 = \re\Bigl(\rme^{\ci q_1} \left[s_1\sigma_1+ s_2\sigma_1 + 
  s_1\sigma_4 \rme^{-\ci \alpha_1}
    + s_2\sigma_4 \rme^{-\ci (\alpha_1+\alpha_3)}
  \right]
  \nonumber \\ & \qquad
  + \rme^{\ci q_2} \left[s_3\sigma_1+ s_4\sigma_1 + s_3\sigma_4 
  \rme^{\ci (\alpha_3-\alpha_2)}
    + s_4\sigma_4 \rme^{-\ci \alpha_2}\right]
  \nonumber \\ & \qquad
  + \rme^{\ci q_3} \bigl[s_1\sigma_2+ s_2\sigma_2\rme^{-\ci \alpha_3}+
  s_3\sigma_2 + s_4\sigma_2\rme^{-\ci \alpha_3}
    \nonumber \\ & \qquad\quad
  + s_1\sigma_3 \rme^{\ci \alpha_1} + s_2\sigma_3\rme^{\ci \alpha_1}+
  s_3\sigma_3 \rme^{\ci (\alpha_2-\alpha_3)} + s_4\sigma_3
  \rme^{\ci (\alpha_2-\alpha_3)}\bigr]
 \Bigr)\, .
\end{align}

This shows that for any fixed $s\in \R^4$ and $\alpha\in [-\pi,\pi]^3$ there
is $\varphi(s,\alpha)\in \R^3$ such that  
\begin{align}
 & g(s,k,k') = \sum_{i=1}^3 R_i(s,\alpha) \cos(q_i+\varphi_i(s,\alpha)) \, ,
\end{align}
where
\begin{align}
 & R_1(s,\alpha) :=  \left|s_1\sigma_1+ s_2\sigma_1 + s_1\sigma_4 
  \rme^{-\ci \alpha_1}
    + s_2\sigma_4 \rme^{-\ci (\alpha_1+\alpha_3)}
  \right|\, , \\ &
   R_2(s,\alpha) :=  \left|s_3\sigma_1+ s_4\sigma_1 
  + s_3\sigma_4 \rme^{\ci (\alpha_3-\alpha_2)}
    + s_4\sigma_4 \rme^{-\ci \alpha_2}\right|\, , \\ &
   R_3(s,\alpha) :=  \bigl|s_1\sigma_2+ s_2\sigma_2\rme^{-\ci \alpha_3}+
  s_3\sigma_2 + s_4\sigma_2\rme^{-\ci \alpha_3}
    \nonumber \\ & \qquad
  + s_1\sigma_3 \rme^{\ci \alpha_1} + s_2\sigma_3\rme^{\ci \alpha_1}+
  s_3\sigma_3 \rme^{\ci (\alpha_2-\alpha_3)} 
  + s_4\sigma_3\rme^{\ci (\alpha_2-\alpha_3)}\bigr|
 \, .
\end{align}
Therefore, we can first integrate over $q$, and this yields
\begin{align}
& F(s) =  \int_{[-\pi,\pi]^3} \frac{\rmd \alpha}{(2\pi)^3} 
  \prod_{i=1}^3 f(R_i(s,\alpha))\, , 
\end{align}
where $f(r) := \int_{-\pi}^\pi \frac{\rmd p}{2\pi} \rme^{-\ci r \cos p}$.  As
shown in \cite{NLS09}, by a direct saddle point argument, here we can find a
pure constant $C_1$ such that $|f(r)|\le C_1 (1+|r|)^{-\frac{1}{2}}$.   

Collecting all the above estimates together, we can conclude that
\begin{align}
 & \int_{\R^4}\! \rmd s |\mathcal{G}(s)| \le 
  \left(\frac{C_1}{2\pi}\right)^9  \int_{\R^4}\! \rmd s
  \left(\int_{[-\pi,\pi]^3}\!\! \rmd \alpha 
  \prod_{i=1}^3 (1+R_i(s,\alpha))^{-\frac{1}{2}}\right)^3
  \, .
\end{align}
By Fubini's theorem, even if infinite, the integral here  is equal to 
\begin{align}\label{eq:G9bound}
 &  \int_{([-\pi,\pi]^3)^3}\!\! \rmd \alpha^{(1)} \rmd \alpha^{(2)}  
  \rmd \alpha^{(3)}  
 \int_{\R^4}\! \rmd s \prod_{j,\ell=1}^3 f_{\ell,j}(s)^{-\frac{1}{4}}\, ,
\end{align}
where for all $i,j\in \set{1,2,3}$ we define
\begin{align}
&   m_{i,j}(s):=1+R_i(s,\alpha^{(j)})\, ,  \\
&   f_{1,j}(s) := m_{1,j}(s) m_{2,j}(s)\, ,\quad
  f_{2,j}(s) := m_{2,j}(s) m_{3,j}(s)\, ,\quad
  f_{3,j}(s) := m_{1,j}(s) m_{3,j}(s)\, .
\end{align}
Thus by H\"older's inequality (\ref{eq:G9bound}) is bounded by 
\begin{align}
 & \int_{([-\pi,\pi]^3)^3}\!\! \rmd \alpha^{(1)} \rmd \alpha^{(2)}  
  \rmd \alpha^{(3)}   \prod_{j,\ell=1}^3 \left(
 \int_{\R^4}\! \rmd s  f_{\ell,j}(s)^{-\frac{9}{4}}\right)^{\frac{1}{9}}
  %  \nonumber \\ & \quad
    = \left[ \int_{[-\pi,\pi]^3}\!\! \rmd \alpha^{(1)}
  \Bigl(\prod_{\ell=1}^3
 \int_{\R^4}\! \rmd s  f_{\ell,1}(s)^{-\frac{9}{4}}\Bigr)^{\frac{1}{9}} \right]^3
 \, .
\end{align}

The main point of introducing the functions $f$ above is that each of them is
a product of two terms, one of which depends only on either $R_1$ or $R_2$.
Since $R_1$ does not depend on $s_3$ or $s_4$, and $R_2$ does not depend on
$s_1$ or $s_2$, it is possible to estimate the four-dimensional integral over
$s$ by a product of two two-dimensional integrals.  Even though the above
integrals might be infinite for some ``bad'' choice of $\alpha:=\alpha^{(1)}$,
we can nevertheless always estimate 
\begin{align}
 & \prod_{\ell=1}^3
 \int_{\R^4}\! \rmd s  f_{\ell,1}(s)^{-\frac{9}{4}}
 % \nonumber \\ & \quad
 \le (G_1(\alpha) G_2(\alpha))^2 G_3(\alpha) G_4(\alpha)
   \, ,
\end{align}
where
\begin{align}
  & G_1(\alpha) :=\int\! \rmd s_1 \rmd s_2 \, 
(1+R_1(s,\alpha))^{-\frac{9}{4}} \, , \\ 
  & G_2 (\alpha) := \int\! \rmd s_3 \rmd s_4 \, 
(1+R_2(s,\alpha))^{-\frac{9}{4}} \, , \\ 
  & G_3 (\alpha) := \sup_{s_3,s_4} \int\! \rmd s_1 \rmd s_2 \,  
(1+R_3(s,\alpha))^{-\frac{9}{4}} \, , \\ 
  & G_4(\alpha) := \sup_{s_1,s_2} \int\! \rmd s_3 \rmd s_4 \,  
(1+R_3(s,\alpha))^{-\frac{9}{4}} \, .
\end{align}

We will next derive upper bounds for $G_i$.  Only the most complicated case,
$G_4$, will be considered here in detail; the other cases can be estimated
similarly. 
We first recall the definition of $R_3$ and the lower bound $|z|\ge \max(|\re
z|,|\im z|)$ valid for any complex $z$.  Since $|\sigma_2+\sigma_3\rme^{\ci
  \alpha_2}|\ge |\sin \alpha_2|$, it is nonzero apart from a set of measure
zero.  Whenever $\sin  \alpha_2\ne 0$, we can estimate 
\begin{align}
 & R_3(s,\alpha)  \nonumber \\
 & \quad =  \bigl|s_1 (\sigma_2 +\sigma_3 \rme^{\ci \alpha_1} ) +
 s_2(\sigma_2\rme^{-\ci \alpha_3}+\sigma_3\rme^{\ci \alpha_1})
 %  \nonumber \\ & \qquad 
 +
  s_3(\sigma_2+\sigma_3 \rme^{\ci (\alpha_2-\alpha_3)})
  + s_4 \rme^{-\ci \alpha_3} (\sigma_2+\sigma_3\rme^{\ci \alpha_2}) 
  \bigr|
   \nonumber \\ & \quad 
   \ge |\sigma_2+\sigma_3\rme^{\ci \alpha_2}| 
  \max\left( \left|s_4+ a(\alpha,s_1,s_2,s_3)\right|, 
   \left|s_3 \im\frac{\sigma_2\rme^{\ci \alpha_3}+
  \sigma_3\rme^{\ci \alpha_2}}{\sigma_2+\sigma_3\rme^{\ci \alpha_2}} +
b(\alpha,s_1,s_2)\right|\right) 
 \, ,
\end{align}
where $a,b\in \R$ depend only on the shown variables.  Here 
\begin{align}
  \im\frac{\sigma_2\rme^{\ci \alpha_3}+\sigma_3
  \rme^{\ci \alpha_2}}{\sigma_2+\sigma_3\rme^{\ci \alpha_2}}
  = \im\frac{\sigma_2\rme^{\ci \alpha_3}
  -\sigma_2}{\sigma_2+\sigma_3\rme^{\ci \alpha_2}}
  = \sigma_2 \sigma_3 \sin \frac{\alpha_3}{2} 
  \im\frac{2 \ci \rme^{\ci (\alpha_3-\alpha_2)/2}}{\rme^{\ci
      \frac{\alpha_2}{2}}+\sigma_2 \sigma_3\rme^{-\ci \frac{\alpha_2}{2}}} \,
  . 
\end{align}
If $\sigma_2 \sigma_3=1$, then $|\sigma_2+\sigma_3\rme^{\ci \alpha_2}|=2
|\cos(\alpha_2/2)|$, and we find  
\begin{align}
 |\sigma_2+\sigma_3\rme^{\ci \alpha_2}| \left|\im\frac{\sigma_2\rme^{\ci
       \alpha_3}+\sigma_3\rme^{\ci \alpha_2}}{\sigma_2+\sigma_3\rme^{\ci
       \alpha_2}}\right| 
 = 2 \left|\sin \frac{\alpha_3}{2}\right| \left|\cos
   \frac{\alpha_3-\alpha_2}{2}\right|\, . 
\end{align}
Else, we have $\sigma_2 \sigma_3=-1$, and thus  $|\sigma_2+\sigma_3\rme^{\ci
  \alpha_2}|=2 |\sin(\alpha_2/2)|$ and  
\begin{align}
 |\sigma_2+\sigma_3\rme^{\ci \alpha_2}| \left|\im\frac{\sigma_2\rme^{\ci
       \alpha_3}+\sigma_3\rme^{\ci \alpha_2}}{\sigma_2+\sigma_3\rme^{\ci
       \alpha_2}}\right| 
 = 2 \left|\sin \frac{\alpha_3}{2}\right| \left|\sin
   \frac{\alpha_3-\alpha_2}{2}\right|\, . 
\end{align}
Therefore, apart from a set of measure zero
\begin{align}\label{eq:R3lb}
 & R_3(s,\alpha) 
 %  \nonumber \\ & \quad 
   \ge \max\left(\left|s'_4+ a'(\alpha,s_1,s_2,s_3)\right|, 
    \left| s'_3 + b'(\alpha,s_1,s_2)\right|\right) 
 \, ,
\end{align}
where $a',b'\in \R$,
$s'_3= 2 \left|\sin(\alpha_3/2)\right| |t((\alpha_3-\alpha_2)/2)|s_3$ and 
$s'_4= 2 |t(\alpha_2/2)| s_4$, with $t(x)=\cos (x)$, if $\sigma_2 \sigma_3=1$,
and $t(x)=\sin (x)$, if $\sigma_2 \sigma_3=-1$. 

To estimate $G_4$, we represent the integrand as the product of its square
roots, apply the first of the lower bounds implied by (\ref{eq:R3lb}) to the
first factor, and the second bound to the second factor. 
Then we integrate over $s_4$ first, change it to $s'_4+a'$, and then change
$s_3$ to $s'_3+b'$.  Since $\int_{\R} \!\rmd r (1{+}|r|)^{-9/8}<\infty$, this
shows that there is a pure constant $C$ such that almost everywhere 
\begin{align}
 & G_4(\alpha) \le \frac{C}{|t_4(\alpha_2/2) 
  \sin(\alpha_3/2) t_4((\alpha_3-\alpha_2)/2) |}\, ,
\end{align}
where $t_4$ denotes either ``$\sin$'' or ``$\cos$'' depending on the sign of
$\sigma_2 \sigma_3$. 

Repeating analogous steps to the other three cases, yields also the following
almost everywhere valid uppers bounds for some pure constant $C$ 
\begin{align}
 & G_1(\alpha) \le \frac{C}{|t_1(\alpha_1/2) 
\sin(\alpha_3/2) t_1((\alpha_1+\alpha_3)/2) |}\, , \label{eq:G1bound}\\
 & G_2(\alpha) \le \frac{C}{|t_2(\alpha_2/2) 
\sin(\alpha_3/2) t_2((\alpha_3-\alpha_2)/2) |}\, , \label{eq:G2bound}\\
 & G_3(\alpha) \le \frac{C}{|t_3(\alpha_1/2) 
\sin(\alpha_3/2) t_3((\alpha_1+\alpha_3)/2) |}\, ,
\end{align}
where each $t_i$, $i=1,2,3$, denotes either ``$\sin$'' or ``$\cos$'' depending
on the value of $\sigma$. 

Now we can collect the bounds together and apply once more H\"older's
inequality to simplify the estimates.  We find that 
\begin{align}
 & \int_{\R^4}\! \rmd s |\mathcal{G}(s)| \le 
\left(\frac{C_1}{2\pi}\right)^9 
\int\! \rmd \alpha\, (G_1G_2)^{\frac{1}{3}}
\int\! \rmd \alpha\, (G_1G_4)^{\frac{1}{3}}
\int\! \rmd \alpha\, (G_2G_3)^{\frac{1}{3}}
  \, .
\end{align}
which is finite, since each of the three integrals is finite.  For instance,
by (\ref{eq:G1bound}) and (\ref{eq:G2bound}), 
\begin{align}
& \int\! \rmd \alpha\, (G_1G_2)^{\frac{1}{3}}
   \nonumber \\ & \quad 
\le C^{\frac{2}{3}} \int\! \rmd \alpha_3\, 
  |\sin(\alpha_3/2)|^{-\frac{2}{3}}
 \nonumber \\ & \qquad \times
\int\! \rmd \alpha_1\,  
 |t_1(\alpha_1/2) t_1((\alpha_1+\alpha_3)/2) | ^{-\frac{1}{3}}
\int\! \rmd \alpha_2\,
|t_2(\alpha_2/2) t_2((\alpha_3-\alpha_2)/2) |^{-\frac{1}{3}}
   \nonumber \\ & \quad 
\le C^{\frac{2}{3}} \int\! \rmd \alpha_3\, |\sin(\alpha_3/2)|^{-\frac{2}{3}}
 %\nonumber \\ & \qquad \times
\int\! \rmd \alpha_1\,  |t_1(\alpha_1/2)| ^{-\frac{2}{3}}
\int\! \rmd \alpha_2\,
|t_2(\alpha_2/2)|^{-\frac{2}{3}}\, ,
\end{align}
and all of the remaining integrals contain only integrable singularities of
the form $r^{-2/3}$, for all allowed choices of $t_1$ and $t_2$.   
The integrals
$\int\! \rmd \alpha\, (G_1G_4)^{\frac{1}{3}}$ and  
$\int\! \rmd \alpha\, (G_2G_3)^{\frac{1}{3}}$ are of the same form, and the
argument proving their finiteness is identical to the above.  This completes
the proof of the Proposition. 
\end{proof}

% \newcommand{\utildir}[1]{../../../texstuff/#1}
% \bibliographystyle{\utildir{abunst_titles}}
% % \bibliographystyle{spmpsci}
% \bibliography{\utildir{myabbr},\utildir{mrabbrev},\utildir{allrefs}}

\begin{thebibliography}{10}

\bibitem{Nord28}
L.~W. Nordheim,
{\it On the kinetic method in the new statistics and its application in the
  electron theory of conductivity\/},
Proc. R. Soc. Lond. Ser. A~{\bf 119} (1928) 689--698.

\bibitem{Peierls29}
R.~Peierls,
{\it {Zur kinetischen {T}heorie der {W}\"armeleitung in {K}ristallen}\/},
Ann. Phys.~{\bf 3} (1929) 1055--1101.

\bibitem{VillaniRev}
C.~Villani,
{\it A review of mathematical topics in collisional kinetic theory\/}.
\newblock In S.~J. Friedlander and D.~Serre (eds.), {\it Handbook of
  Mathematical Fluid Dynamics\/}, volume~1. Elsevier, Amsterdam, 2002.

\bibitem{dolb94}
J.~Dolbeault,
{\it Kinetic models and quantum effects: {A} modified {B}oltzmann equation for
  {F}ermi-{D}irac particles\/},
Arch. Ration. Mech. Anal.~{\bf 127} (1994) 101--131.

\bibitem{Lions94iii}
P.-L. Lions,
{\it Compactness in {B}oltzmann's equation via {F}ourier integral operators and
  applications, {III}\/},
J. Math. Kyoto Univ.~{\bf 34} (1994) 539--584.

\bibitem{Lu04}
X.~Lu,
{\it On isotropic distributional solutions to the {B}oltzmann equation for
  {B}ose-{E}instein particles\/},
J. Stat. Phys.~{\bf 116} (2004) 1597--1649.

\bibitem{Lu05}
X.~Lu,
{\it The {B}oltzmann equation for {B}ose-{E}instein particles: Velocity
  concentration and convergence to equilibrium\/},
J. Stat. Phys.~{\bf 119} (2005) 1027--1067.

\bibitem{EMV08}
M.~Escobedo, S.~Mischler, and J.~J.~L.~Vel\'{a}zquez,
{\it Singular solutions for the {U}ehling-{U}hlenbeck equation\/},
Proc. Roy. Soc. Edinburgh Sect. A~{\bf 138} (2008) 67--107.

\bibitem{EV12a}
M.~{Escobedo} and J.~J.~L. {Vel{\'a}zquez},
{\it Finite time blow-up for the bosonic {N}ordheim equation\/},
preprint (2012), 53 pp. URL \url{http://arxiv.org/abs/1206.5410}.

\bibitem{EV12b}
M.~{Escobedo} and J.~J.~L. {Vel{\'a}zquez},
{\it On the blow up of supercritical solution of the {N}ordheim equation for
  bosons\/},
preprint (2012), 29 pp. URL \url{http://arxiv.org/abs/1210.1664}.

\bibitem{flms12}
M.~L.~R. F\"{u}rst, J.~Lukkarinen, P.~Mei, and H.~Spohn,
{\it Derivation of a matrix-valued {B}oltzmann equation for the {H}ubbard model\/},
preprint (2013), 18 pp. URL \url{http://arxiv.org/abs/1306.0934}.

\bibitem{NLS09}
J.~Lukkarinen and H.~Spohn,
{\it Weakly nonlinear {S}chr{\"o}dinger equation with random initial data\/},
Invent. Math.~{\bf 183} (2011) 79--188.

\bibitem{rudin:fa}
W.~Rudin,
{\it Functional Analysis\/}.
\newblock Tata McGraw-Hill, New Delhi, 1974.

\bibitem{AY81}
H.~Araki and S.~Yamagami,
{\it An inequality for {H}ilbert-{S}chmidt norm\/},
Commun. Math. Phys.~{\bf 81} (1981) 89--96.

\bibitem{SeSi75}
E.~Seiler and B.~Simon,
{\it An inequality among determinants\/},
Proc. Nat. Acad. Sci. USA~{\bf 72} (1975) 3277--3278.

\bibitem{gohberg00}
I.~Gohberg, I.~Goldberg, and N.~Krupnik,
{\it Traces and Determinants of Linear Operators\/}.
\newblock Birkh{\"a}user, Berlin, 2000.

\bibitem{reedsimonII}
M.~Reed and B.~Simon,
{\it Methods of Modern Mathematical Physics {II}: {F}ourier Analysis,
  Self-Adjointness\/}.
\newblock Academic Press, New York, 1975.

\bibitem{Bellman70}
R.~Bellman,
{\it Introduction to Matrix Analysis\/}.
\newblock SIAM, Philadelphia, 2nd edition, 1970.

\bibitem{rudin:rca}
W.~Rudin,
{\it Real and Complex Analysis\/}.
\newblock McGraw-Hill, New York, 3rd edition, 1987.

\end{thebibliography}
% \end{document}

\end{document}